\documentclass[a4paper,draft,11pt]{article}

\usepackage{pb-diagram,amsmath,amssymb}
\usepackage{bbm}
\usepackage{dsfont}
\usepackage{amsthm}
\usepackage[T2A]{fontenc}
\usepackage[cp1251]{inputenc}
\usepackage[active]{srcltx}

\usepackage{amsthm}

\newtheorem{theorem}{Theorem}
\newtheorem{definition}[theorem]{Definition}
\newtheorem{lemma}[theorem]{Lemma}
\newtheorem{lemmadefinition}[theorem]{Lemma-Definition}
\newtheorem{propositiondefinition}[theorem]{Proposition-Definition}
\newtheorem{proposition}[theorem]{Proposition}
\newtheorem{corollary}[theorem]{Corollary}

\newcommand{\CC}{\mathbb C}

\newcommand{\FF}{\mathbb F}

\newcommand{\NN}{\mathbb N}

\newcommand{\PP}{\mathbb P}
\newcommand{\QQ}{\mathbb Q}
\newcommand{\RR}{\mathbb R}

\newcommand{\ZZ}{\mathbb Z}
\newcommand{\cA}{\mathcal A}

\newcommand{\cE}{\mathcal E}

\newcommand{\cL}{\mathcal L}
\newcommand{\cM}{\mathcal M}

\newcommand{\cO}{\mathcal O}

\newcommand{\cW}{\mathcal W}

\def\myarrow{\ \hbox to 2em{\leaders
\hbox to 0.5ex{\hss\raise 0.55ex\hbox to 0.3ex{\hrulefill}\hss}
\hfill\,\llap{$>$}}\ }

\title{ Mac Williams identities  and  polarized  Riemann-Roch  conditions
 \footnote{{\it  2010 Mathematics Subject Classification:} Primary: 94B27,  14G50; Secondary:    11T71.   \protect\\
 {\it Key words and phrases:}
 Mac Williams identities,
 Duursma's reduced polynomial of an additive code,
 Polarized Riemann-Roch Conditions \protect\\
Supported by   Contract    57/12.04.2016   with the Scientific Foundation of Kliment Ohridski  University of Sofia. } }
\author{   Azniv Kasparian, Ivan Marinov }
\date{      }

\begin{document}
\maketitle

\thispagestyle{empty}

\begin{abstract}
 The present note establishes the equivalence of Mac Williams identities for an additive code $C$ and its dual $C^{\perp}$ to Polarized Riemann-Roch Conditions on their $\zeta$-functions.
 In such a way, the duality of additive codes appears  to be a  polarized form of the Serre duality on a smooth irreducible projective curve.
\end{abstract}

%%%%%%%%%%%%%%%%%%%%%%%%%%%%%%%%%%%%%%%%%%%%%%%%%%%%%%%%%%%%%%%%%%%%%%%%%%%%%%%%%%%%%%%%%%%%%%%%%%%%%%%%%%%%%%%%%%%%%%%%%%%%%%%%%%%%%%%%%%%%%%%%%%%%%%%%%%%%%
\section{ Introduction }
 %%%%%%%%%%%%%%%%%%%%%%%%%%%%%%%%%%%%%%%%%%%%%%%%%%%%%%%%%%%%%%%%%%%%%%%%%%%%%%%%%%%%%%%%%%%%%%%%%%%%%%%%%%%%%%%%%%%%%%%%%%%%%%%%%%%%%%%%%%%%%%%%%%%%%%%%%%

Let $(G, +)$ be a finite abelian group and $(\widehat{G}, .)$ be the group of the multiplicative characters $\pi : (G, +) \rightarrow (\CC ^*, .)$ of $G$.
The subgroups $(C, +)$ of $(G^n, +)$ are called additive codes.
Any linear code $C \subset \FF_q ^n$ over a finite field $\FF_q$ is an additive code in the $n$-th Cartesian power of the finite abelian group
$(\FF_q, +) \simeq (\ZZ_p ^m, +)$ with $p = {\rm char} \FF_q$, $q = p^m$.
The dual code
$$
C^{\perp} := \{ \pi = (\pi _1, \ldots , \pi _n) \in \widehat{G}^n \, \vert \, \pi ( a) = 1, \, \forall a \in C   \}
$$
is a subgroup $(C^{\perp}, .)$ of $(\widehat{G}^n, .) = (\widehat{G^n}, .)$ and can be viewed as an additive code over $\widehat{G}$.
If $\varepsilon: G \rightarrow \CC^*$ is the trivial character with $\varepsilon (g) = 1$ for $\forall g \in G$ then the Hamming weight on $G$ and $\widehat{G}$ are defined as
\begin{equation}    \label{HammingWeightG}
{\rm wt} : G \longrightarrow  \{ 0,1 \}, \ \ {\rm wt} (g) := \begin{cases}
0   &  \text{ if $g = 0_G$, }  \\
1  &  \text{ if $ g \neq 0_G$,  }
\end{cases}
\end{equation}
respectively,
\begin{equation}    \label{HammingWeightGPerp}
{\rm wt} : \widehat{G}  \longrightarrow  \{ 0,1 \}, \quad  {\rm wt} ( \pi) :=  \begin{cases}
0   &  \text{ if $\pi = \varepsilon$,  } \\
1  &  \text{ if $\pi \neq \varepsilon$.  }
\end{cases}
\end{equation}
For an arbitrary $n \in \NN$, these extend to
\begin{equation}    \label{HammingWeightGn}
{\rm wt} : G^n \longrightarrow \{ 0, 1, \ldots , n \}, \quad {\rm wt} (a_1, \ldots , a_n) := \sum\limits _{i=1} ^n {\rm wt} (a_i), \ \ \mbox{  respectively}
\end{equation}
\begin{equation}    \label{HammingWeightGPerpN}
{\rm wt} : \widehat{G} ^n \longrightarrow  \{ 0, 1, \ldots , n \}, \quad {\rm wt} ( \pi _1, \ldots , \pi _n) := \sum\limits _{i=1} ^n {\rm wt} ( \pi _i).
\end{equation}
That enables to define the homogeneous weight enumerator
$$
\cW_C(x,y) := \sum\limits _{c \in C} x^{n - {\rm wt} (c)} y^{{\rm wt} (c)}
$$
of an additive code $(C, +) \leq (G^n , +)$.
Similarly, the dual code
$$
C^{\perp} := \{  \pi = (\pi _1, \ldots , \pi _n) \in \widehat{G}^n \, \vert \,  \pi (a) = 1, \forall a \in C \} \simeq (\widehat{ G^n  /  C },.)
$$
has homogeneous weight enumerator
$$
\cW_{C^{\perp}} := \sum\limits _{\pi \in C^{\perp}} x^{n - {\rm wt} ( \pi)} y^{{\rm wt} (\pi)}.
$$
According to \cite{Delsarte} (see also \cite{Wood} or \cite{Heide}), Fourier inversion formula for the function
$$
F: (G^n/C, +)  \longrightarrow \CC [x,y] ^{(n)},
$$
associating to a coset $a  + C \in (G^n/C,+)$ its homogeneous weight enumerator $F(a + C)$ provides Mac Williams identities
$$
\cW_{C^{\perp}} (x,y) = \frac{1}{|C|}  \cW_C (x + (|G|-1)y, x-y)
$$
for $\cW_C(x,y)$, $\cW_{C^{\perp}} (x,y)$.

Mac Williams initiates the study of the duality of linear codes by their weight distributions with respect to the Hamming weight in \cite{OriginalMW}.
Delsarte generalizes Mac Williams  results in \cite{Delsarte6} by the means of association schemes.
Zinoviev and Ericson's  \cite{ZinovievEricson} describes Mac Williams duality for additive codes $(C, +) \leq (G^n, +)$ and their duals $(C^{\perp}, .) \leq (\widehat{G}^n, .)$ with respect to isomorphic partitions of $G^n$ and $\widehat{G}^n$.
In \cite{Heide}  Gluering-Luerssen  proves Mac Williams identities for $(C^{\perp}, .) \leq (\widehat{G}^n, .)$ and $(C, +) \leq (G^n, +)$ with respect to an arbitrary partition with $M$ blocks  on $\widehat{G}^n$ and its Fourier transform, which is a partition on $G^n$.
In the case of $M=2$, her set up reduces to the Hamming weights on $G, \widehat{G}$ and $G^n, \widehat{G}^n$.
Mac Williams identities over finite Frobenius rings are studies by Greferath-Schmidt's \cite{GreferathSchmidt}, Honold-Landjev's \cite{HonoldLandjev}, Wood's \cite{Wood}, etc.

Let $C \subset \FF_q ^n$ be an $\FF_q$-linear $[n,k,d]$-code with dual
$$
C^{\perp} := \{ a = (a_1, \ldots , a_n) \in \FF_q ^n \, \vert \, \langle a, c \rangle = \sum\limits _{i=1} ^n a_i c_i =0, \ \ \forall c \in C \}
$$
of minimum distance $d^{\perp}$.
The deviation $g := n+1 -d-k \in \ZZ^{\geq 0}$ of the parameters of $C$ from the equality in the Singleton bound is called the genus of $C$.
In \cite{D1}, \cite{D2} Duursma introduces the  $\zeta$-polynomials $P_C(t), P_{C^{\perp}}(t) \in \QQ[t]$  of degree $\deg P_C(t) = \deg P_{C^{\perp}} (t) = g + g^{\perp} = n+2 - d - d^{\perp}$ and shows that Mac Williams identities  for $C, C^{\perp}$ are equivalent  to the functional equation
\begin{equation}     \label{MWDzetaPolynomials}
P_{C^{\perp}} (t) = P_C \left( \frac{1}{qt} \right) q^g t^{g + g^{\perp}}
\end{equation}
for their $\zeta$-polynomials.
Note that (\ref{MWDzetaPolynomials}) is a polarized form of the functional equation of the Hasse-Weil polynomial of a smooth irreducible projective curve of genus $g$, defined over $\FF_q$.
 The article \cite{PShW} of Pellikaan, Shen and van Wee sheds a light on this phenomenon.
 More precisely, \cite{PShW} shows that for an arbitrary $\FF_q$-linear code $C \subset \FF_q ^n$ there is a smooth irreducible projective curve $X / \FF_q \subset \PP^N ( \overline{\FF_q})$, defined over $\FF_q$, distinct rational points $P_1, \ldots , P_n \in X( \FF_q) := X \cap \PP^N ( \FF_q)$ and a divisor $G$ of $\FF_q(X)$, whose support is disjoint from the support of $D = P_1 + \ldots + P_n$, such that $X = \cE_D H^0 (X, \cO _X ( [G]))$ coincides with the image of the evaluation map
 $$
 \cE_D : H^0 (X, \cO _X ([G])) = \cL _X (G) \longrightarrow \FF_q ^n, \ \ \cE_D (f) = (f(P_1), \ldots , f(P_n)) \ \ \mbox{  for  } \ \ \forall f \in \cL _X(G).
 $$
 The kernel of $\cE_D$ coincides with $\cL_X (G-D)$ (cf. Proposition \ref{ProjInterpretationOrbitsOfEffectiveDivisors} (i) from section 5) and $C$ is isomorphic to the quotient space $\cL _X (G) / \cL _X (G-D)$ as a linear space over $\FF_q$.
 The dual code $C^{\perp}$ is isomorphic to the quotient space $\cL _X (K_X - G +D) / \cL _X(K_X-G)$, where $K_X$ stands for a canonical divisor of $X$.
 Under the Serre duality on $X$, the first cohomology group $H^1 (X, \cO _X ([E]))$ is isomorphic to the sections $\cL _X (K_X - E) = H^0 (X, \cO _X ( [K_X - E]))$ of the line bundle $\cO _X([E])$, corresponding to the divisor $K_X - E$.
 If we view the divisor $K_X - E$ as a Serre dual of $E$ then the presentations $C \simeq \cL _X(G) / \cL _X (G-D)$ and $C^{\perp} \simeq \cL _X (K_X - G +d) / \cL _X (K_X - G)$ of mutually dual linear codes are compatible with the Serre duality on $X$.
 From now on, let $l(E) : \dim _{\FF_q} \cL _X (E)$ be the dimension of the space $\cL _X(E) = H^0 (X, \cO _X([E]))$ of the global sections of $\cO _X ([E])$.
 The Riemann-Roch Theorem on $X$ is a numerical expression of the difference $l(G) - l (K_X - G)$ by topological invariants of $X, G$, i.e., by the genus $g$ of $X$ and the degree $m$ of $G$.
 Thus, it is reasonable the numerical relation between the weight distributions of $C$, $C^{\perp}$, provided by Mac Williams identities to be compatible with the Serre duality on $X$ and to play the role of the Riemann-Roch Theorem for $C$, $C^{\perp}$.

 The relation between the local Weil $\zeta$-function $\zeta _X(t)$ of $X$ and  Duursma's  $\zeta$-functions $\zeta _{C_i} (t) := \frac{P_{C_i} (t)}{(1-t)(1-qt)}$ of the linear codes $C_i = \cE _D \cL _X (G_i)$, associated with a complete set of representatives $G_i$, $1 \leq i \leq h$  of the linear equivalence classes of the divisors of $\FF_q (X)$ of degree $2g \leq m < n$ is noticed by Duursma in \cite{D1}, \cite{D2}.
 However, the algebraic-geometric representations $C = \cL _X (G)$ of arbitrary linear codes $C \subset \FF_q ^n$, constructed by Pellikaan, Shen and van Wee in \cite{LShW} tend to have $g > m >n$.
 As a result, if there exist $G_i$ with ${\rm Supp} (G_i) \cap {\rm Supp} (D) = \emptyset$ for $\forall 1 \leq i \leq h$ then the $\zeta$0functions of $C_i = \cE _D \cL _X (G_i)$ are related with the truncated local Weil $\zeta$-function $\zeta _X^{(m)} (t) = \frac{ P_X^{(m)} (t)}{(1-t)(1-qt)}$ of $X$.
 (If $P_X(t) \in \ZZ[t]$ is the Hasse-Weil polynomial of $X$ of degree $2g$ then $P_X^{(m)} (t)$ is the sum of the terms of $P_X(t)$ of degree $\leq m$.)  
Besides, if a linear equivalence class $G_i + {\rm div} \FF_q (X)$ of divisors of $\FF_Q (X)$ of degree $m$ has no representative $G_i$ with ${\rm Supp} (G_i) \cap {\rm Supp} (D) = \emptyset$ then the  evaluation map $\cE_D$ at $D$ does not act on $G_i + {\rm div} \FF_q (X)$ and the available $\zeta$-functions of algebraic-geometric codes do not reflect the information for the effective divisors from $G_i + {\rm div} \FF_q  (X)$.
For a detailed discussion of this kind of problems see \cite{K4}.

The aim of the present note is to understand Mac Williams duality of additive codes in terms of algebraic geometry.
It is completely independent and of a different, more formal nature than   the recent work \cite{Hugues} of  Randriambololona.
Note that the   Riemann-Roch Theorem 44  for linear codes $C, C^{\perp} \subset \FF_q^n$  from \cite{Hugues}   is stronger than  our Polarized Riemann-Roch Conditions ${\rm PRRC} _q (g,g^{\perp})$  on $\zeta _C(t)$, $\zeta _{C^{\perp}}(t)$, as far as it   implies the functional equation on $\zeta _C(t), \zeta _{C^{\perp}}(t)$,    which we show to be equivalent to ${\rm PRRC}_q (g, g^{\perp})$.
Mac Williams identities are used in Delsarte's \cite{Delsarte6}, Byrne-Greferath-Sullivan's \cite{ByrneGreferathSullivan} and other works for obtaining linear programming bounds on codes.
For applications in the engineering one can see ElKhamy-McEliece's  \cite{KhamyMcEliece} or Lu-Kumar-Yang's \cite{LuKumarYang}.

The main result of the present article is the equivalence  of Mac Williams identities for additive codes $(C,+) \leq (G^n, +)$, $(C^{\perp}, .) \leq ( \widehat{G}^n, .)$  to Polarized Riemann-Roch Conditions for their $\zeta$-functions
$$
\zeta _C(t) := \frac{P_C(t)}{(1-t)(1-|G|t)}, \ \ \zeta _{C^{\perp}} (t) := \frac{P_{C^{\perp}} (t)}{(1-t)(1-|G|t)}
$$
In such a way, Mac Williams duality of additive codes turns to be a polarized form of Serre duality from algebraic geometry.
A crucial step from the proof of the aforementioned equivalence is the study of the additive MDS-codes $({\rm MDS}(n,d), +) \lneq (G^n, +)$ of length $n$ and minimum distance $d$, defined as the ones of genus $g= n+1 -d - \log _{|G|} (|{\rm MDS} (n,d)|)=0$.
After showing that for any $(n-k)$-tuple of indices $\beta = \{ \beta _1, \ldots , \beta _{n-k} \} \varsubsetneq [n] :=  \{ 1, \ldots , n \}$ the puncturing (or erasing) $\Pi _{\beta} : ({\rm MDS} (n,d), +) \rightarrow (G^k, +)$ of the components, labeled by $\beta$ is an isomorphism, we compute explicitly the homogeneous weight enumerator $\cM _{n,d} (x,y)$ of ${\rm MDS} (n,d)$ and observe that it  depends only on $n$ and $d$.

Here is a synopsis of the article.
In section 2 we study the additive codes of genus $0$, called the additive MDS-codes.
After showing that the dual of an additive MDS-code ${\rm MDS} (n,d)$ of length $n$ and minimum distance $d>1$ is an additive MDS-code ${\rm MDS} (n, n+2-d)$ of length $n$ and minimum distance $n+2-d$, we establish that for any unordered $d$-tuple $\gamma \in \binom{[n]}{d}$ with entries from $[n]$ there are exactly $|G|-1$ words of ${\rm MDS} (n,d)$ with support $\gamma$.
Then we show that the shortening of the dual $(C^{\perp},.) \leq (\widehat{G}^n, .)$ of an arbitrary additive code $(C,+) \leq (G^n, +)$ at some component coincides  with the puncturing of $C$ at that component.
That enables to obtain explicitly the number $\cM _{n,d} ^{(s)}$ of the words of weight $d \leq s \leq n$ in an additive MDS-code ${\rm MDS} (n,d)$.
The third section introduces the $\zeta$-polynomial $P_C(t) \in \QQ[t]$ and Duursma's reduced polynomial $D_C(t) \in \QQ[t]$ of an additive code $(C, +) < (G^n, +)$ and expresses Mac Williams identities for $C, C^{\perp}$ as functional equations on $P_C (t), P_{C^{\perp}} (t)$ or, respectively, on $D_C(t), D_{C^{\perp}} (t)$.
The fourth section expresses the Riemann-Roch Theorem on a smooth irreducible projective curve $X$ of genus $g \geq 0$, defined over a finite field $\FF_q$ as (non-polarized) Riemann-Roch Conditions with base $q \in \NN$ and genus $g$ on the local Weil $\zeta$-function $\zeta _X(t)$ of $X$.
That motivates the notion of Polarized Riemann-Roch Conditions ${\rm PRRC} _q (g, g^{\perp})$ with base $q \in \NN$ and genera $g, g^{\perp} \in \ZZ ^{\geq 0}$ on a pair $\zeta (t), \zeta ^{\perp} (t) \in \CC [[t]]$ of formal power series in one variable $t$.
The functional equation for $D_C(t), D_{C^{\perp}} (t) \in \QQ[t]$, expressing Mac Williams identities for the weight distribution of $C, C^{\perp}$ of genera $g, g^{\perp}$ is shown to be equivalent to the Polarized Riemann-Roch Conditions ${\rm PRRC} _{|G|} (g, g^{\perp})$ on the $\zeta$-functions $\zeta _C(t)$, $\zeta _{C^{\perp}} (t)$.
As a consequence, the lower parts $\varphi _C (t) = \sum\limits _{i=0} ^{g-2} c_i t^{i} \in \QQ[t]$, $\varphi _{C^{\perp}} (t) = \sum\limits _{i=0} ^{g^{\perp} -2} c_i ^{\perp} t^{i} \in \QQ[t]$ of Duursma's reduced polynomials $D_C (t) = \sum\limits _{i=0} ^{g + g^{\perp} -2} c_i t^{i}$, $D_{C^{\perp}} (t) = \sum\limits _{i=0} ^{g + g^{\perp} -2} c_i ^{\perp} t^{i} \in \QQ[t]$ and the number $c_{g-1} = c^{\perp}_{g^{\perp}-1} \in \QQ$ turn to determine completely $D_C(t)$, $D_{C^{\perp}}(t)$.
The final fifth section discusses some averaging, algebraic-geometric  and probabilistic interpretations of the coefficients $c_i \in \QQ$ of Duursma's reduced polynomial $D_C (t) = \sum\limits _{i=0} ^{g + g^{\perp} -2} c_i t^{i} \in \QQ[t]$ of an additive code $(C,+) < (G^n, +)$, with a specific emphasis on the case of an $\FF_q$-linear code $C \subset \FF_q ^n$, when $C$ is not only a subgroup of $(\FF_q ^n, +)$ but also a subset with an $\FF_q^*$-action $\FF_q ^* \times C \rightarrow C$, $(\lambda ,a) \mapsto (\lambda a_1, \ldots , \lambda a_n)$, preserving the Hamming weight.
In general, $(|G|-1) c_i$ with $0 \leq i \leq g-1$ is shown to be the  average coordinates of an intersection of $C \setminus \{ 0_G^n \}$ with $n-d-i$ coordinate hyperplanes in $(G^n, +)$.
In the case of an $\FF_q$-linear code $C \subset \FF_q ^n$, the presence of an algebro-geometric realization $C = \cE _D \cL _X (G)$, established by Pellikaan, Shen and van Wee in \cite{PShW} allows to inter[ret the projectivization $\PP(C)$  as an $\cL _X(G-D)$-orbit space of an explicit finite set of effective divisors of $\FF_q (X)$.
As a result, the coefficients of $\zeta _C(t) = \sum\limits _{i=0} ^{\infty} \cA _i (C) t^{i}$ for $0 \leq i \leq g-1$ are shown  to be average cardinalities of appropriate $\cL _X (G-D)$-orbit spaces of effective divisors of $\FF_q (X)$.
In particular, $\binom{n}{d+i} \cA _i (C) \in \ZZ ^{\geq 0}$ for $\forall 0 \leq i \leq g + g^{\perp} -2$ and $\cA _i (C) \in \ZZ ^{\geq 0}$ for $\forall i > g + g^{\perp} -2$.
As a result,  Tsfasman-Vl${\rm \breve{a}}$dut-Nogin's coefficients $B_{d+i}$, $0 \leq i \leq g-1$  from
$\cW_C (x,y) = x^n + \sum\limits _{i=0} ^{n-d} B_i (x-y)^{i} y^{n-i}$, given in \cite{TVN},  turn to give the same information as the coefficients $c_i$ of Duursma's reduced polynomial $D_C(t)$, due to  $B_{d+i} = \binom{n}{d+i} (q-1)c_i$ for  $0 \leq i \leq g-1$.
We express $c_i$ with $0 \leq i \leq g-1$ by the probabilities of $a \in G^n$ of weight $d \leq {\rm wt} (a) = s \leq d+i$ to belong to $C$.
Similarly, $c_i$ with $g \leq i \leq g + g^{\perp} -2$ are related to the probabilities of $\pi \in \widehat{G}^n$ of weight $d^{\perp} \leq {\rm wt} ( \pi) = s \leq n-d-i$ to belong to $C^{\perp}$.
Finally, the sum of the probabilities $\overline{p}_a ^{(d+i)}$ of a $(d+i)$-tuple of indices to contain the support of some $a \in C \setminus \{ 0_G ^n \}$ is shown to be $(|G|-1) c_i$ for $0 \leq i \leq g-1$, while the sum of the probabilities $\overline{p} _{\pi} ^{(n-d-i)}$ of an $(n-d-i)$-tuple of indices to contain $\pi \in C^{\perp} \setminus \{ \varepsilon \}$ turns to be $(|G|-1) c_i |G| ^{g-i+1}$ for $g \leq i \leq g^{\perp} -2$.
In the case of $\FF_q$-linear codes, the factor $|G|-1$ disappears by replacing $a \in C \setminus \{ 0_{\FF_q}^n \}$ with $[a] \in \PP(C) \subset \PP ( \FF_q ^n)$ and $\pi \in C^{\perp} \setminus \{ 0_{\FF_q^n} \}$ with $[\pi] \in \PP (C^{\perp}) \subset \PP ( \FF_q ^n)$.

%%%%%%%%%%%%%%%%%%%%%%%%%%%%%%%%%%%%%%%%%%%%%%%%%%%%%%%%%%%%%%%%%%%%%%%%%%%%%%%%%%%%%%%%%%%%%%%%%%%%%%%%%%%%%%%%%%%%%%%%%%%%%%%%%%%%%%%%%%%%%%%%%%%%%%%%%%%%%
\section{ Additive MDS weight enumerators }
 %%%%%%%%%%%%%%%%%%%%%%%%%%%%%%%%%%%%%%%%%%%%%%%%%%%%%%%%%%%%%%%%%%%%%%%%%%%%%%%%%%%%%%%%%%%%%%%%%%%%%%%%%%%%%%%%%%%%%%%%%%%%%%%%%%%%%%%%%%%%%%%%%%%%%%%%%%

If $(G,+) \neq \{ 0_G \}$ is an additively  written   non-zero finite abelian group then the homomorphisms $\pi : (G, +) \rightarrow ( \CC^*, .)$ are called multiplicative characters of $G$.
If $G$ is of order $m$ then $\pi (G)$ consists of   $m$-th roots of unity  and, in particular, $\pi$ maps to the subgroup $(S^1,.)$ of $(\CC^*, .)$, supported by the unit circle $S^1 := \{ z \in \CC \, \vert \, z \overline{z} =1 \}$.
The set $\widehat{G}$ of the multiplicative characters of $G$ is a group with respect to the pointwise multiplication
$$
\chi \pi : G \longrightarrow S^1, \quad (\chi \pi) (g) := \chi (g) \pi (g) \ \ \mbox{  for } \ \ \forall g \in G, \ \ \forall \chi, \pi \in \widehat{G}.
$$
The neutral element of this group is the trivial character
$$
\varepsilon : G \longrightarrow \{ 1 \},  \ \ \varepsilon (g) =1 \ \ \mbox{  for } \ \ \forall g \in G.
$$
For an arbitrary $n \in \NN$, the subgroups $(C,+)$ of $(G^n, +)$ are called additive codes of length $n$.
With respect to the Hamming weight on $G^n$, defined by (\ref{HammingWeightG}) and (\ref{HammingWeightGn}), there is a unique word $0_G ^n \in G^n$ of weight $0$, which belongs to any additive code $(C,+) < (G^n, +)$.
The minimal non-zero weight
$$
d  := \min \{ {\rm wt} (c) \in \NN \, \vert \, c \in C \setminus \{ 0_G ^n \} \}
$$
of a word of $C \neq \{ 0_G ^n \}$ is called the minimum weight of $C$ or the minimum distance of $C$.
If $C = \{ 0_G ^n \}$ is the zero code, we assume that $d = 0$.
As far as the Hamming distance
$$
d: C \times C \longrightarrow \ZZ ^{\geq 0}, \quad d(a,b) := {\rm wt} (a-b)
$$
is a metric, the decoding of an additive code of minimum weight $d \in \NN$ with at most $\left[  \frac{d-1}{2} \right]$ perturbed symbols is unique.

Here is a simple lemma on the puncturing and shortening of additive codes and their duals.
The puncturing $\Pi _i  : (C, +) \rightarrow (G^{n-1}, +)$ of an additive code $(C, +) \leq (G^n, +)$ at the $i$-th component is the group homomorphism, deleting the $i$-th component of each word $c \in C$, $\Pi _i (c_1, \ldots , c_n) = (c_1, \ldots , c_{i-1}, c_{i+1}, \ldots , c_n)$.
 By its very definition, the shortening $S_i : (C, +) \rightarrow (G^{n-1}, +)$ at the $i$-th component does not act on $c \in C$ with $c_i \neq 0$ and reduces to the puncturing $\Pi _i$ on the words $c \in C$ with $c_i =0$.
The statement and the proof are the same as for $\FF_q$-linear codes, as exposed in \cite{HuffmanPless}.
We give the proof for completeness.

\begin{lemma}   \label{PuncturingAnd Shortening}
Let $(C, +) \leq (G^n, +)$ be an additive code with dual $(C^{\perp}, .) \leq ( \widehat{G}^n, .)$.
 Denote by  $S_i : C \rightarrow G^{n-1}$, respectively, $S_i : C^{\perp} \rightarrow \widehat{G}^{n-1}$   the shortenings at the   component, labeled by some
  $1 \leq i \leq n$ and put  $\Pi _i : (C, +) \rightarrow (G^{n-1}, +)$, respectively, $\Pi _i : (C^{\perp}, .) \rightarrow ( \widehat{G}^{n-1}, .)$  for  the puncturings at $i$.
Then
\begin{equation}   \label{PuncturingAndShortening}
S_i (C^{\perp}) = \Pi _i (C) ^{\perp} \ \ \mbox{\rm   and  } \ \ \Pi _i (C^{\perp}) = S_i (C) ^{\perp}.
\end{equation}
\end{lemma}

\begin{proof}

Towards the inclusion $S_i (C^{\perp}) \subseteq \Pi _i (C)^{\perp}$, it suffices to note that for an arbitrary $\pi \in C^{\perp}$ with $\pi _i = \varepsilon$ and an arbitrary $c \in C$, one has
$$
S_i ( \pi) ( \Pi _i (c)) = \prod\limits _{j \neq i} \pi _j (c_j) = \left[ \prod\limits _{j \neq i} \pi _j ( c_j) \right] \pi _i ( c_i) = \prod\limits _{j=1} ^n \pi _j ( c_j) = \pi (c) =1.
$$
For the opposite inclusion $\Pi _i (C) ^{\perp} \subseteq S_i ( C^{\perp})$, let us choose $\chi \in (\Pi _i (C)^{\perp},.) \leq ( \widehat{G} ^{n-1}, .)$ and extend it to $\pi \in \widehat{G}^n$ with $\pi _j := \chi _j$ for $1 \leq j \leq i-1$, $\pi _i := \varepsilon$ and $\pi _j := \chi _{j-1}$
 for $i +1 \leq j \leq n$.
Then $\pi \in C^{\perp}$, according to
$$
1 = \chi ( \Pi _i (c)) = \left[ \prod\limits _{j=1} ^{i-1} \chi _j ( c_j) \right] \left[ \prod\limits _{j = i+1} ^n \chi _{j-1} ( c_j) \right] =
\left[ \prod\limits _{j \neq i} \pi _j ( c_j) \right] \varepsilon ( c_i) = \pi (c),
$$
for $\forall c \in C$.
Thus,  $\chi = S_i ( \pi ) \in S_i (C^{\perp})$ and $\Pi _i  (C^{\perp}) \subseteq S_i (C^{\perp})$.
That justifies the coincidence $S_i (C^{\perp}) = \Pi _i (C)^{\perp}$.

In order to check that $\Pi _i (C^{\perp})   \subseteq S_i (C)^{\perp}$, let $\pi \in C^{\perp}$ and $c \in C$ with $c_i = 0_G$.
Then
$$
\Pi _i ( \pi) (S_i (c)) =
 \prod\limits _{j \neq i} \pi _j ( c_j) =
  \left[ \prod\limits _{j \neq i} \pi _j ( c_j) \right]  \pi _i ( 0_G ) =
\prod\limits _{j=1} ^n \pi _j ( c_j) = \pi (c) =1
$$
 reveals that $\Pi _i ( \pi) \in S_i (C)^{\perp}$.
Towards the coincidence $\Pi _i (C^{\perp}) = S_i (C)^{\perp}$, let us note that the application of $S_i (C^{\perp}) = \Pi _i (C)^{\perp}$ to $C^{\perp}$ provides $S_i (C) \simeq S_i ( (C^{\perp}) ^{\perp}) = \Pi _i (C^{\perp}) ^{\perp}$, after combining with the natural group isomorphism $C \simeq (C^{\perp}) ^{\perp}$.
Taking the duals of both sides, one concludes that $S_i (C)^{\perp} \simeq \left[ \Pi _i (C^{\perp}) ^{\perp} \right] ^{\perp} \simeq \Pi _i (C^{\perp})$, whereas $\left| S_i (C) ^{\perp} \right| = \left| \Pi _i (C^{\perp}) \right|$ and $\Pi _i (C^{\perp}) = S_i (C)^{\perp}$.

\end{proof}

The next elementary lemma defines the genus of an additive code and reminds the Singleton Bound for additive codes.
The proof is elementary and coincides with the one for the Singleton Bound of linear codes over finite fields.
We provide it for completeness, as far as we have not found an available reference on it.

\begin{lemmadefinition}    \label{NonNegativeGermsGenus}
Any additive code $(C, +) \leq (G^n, +)$ of minimum distance $d \in \ZZ^{\geq 0}$ over a non-trivial finite abelian  group $G  \neq 0_G$ has non-negative genus
 $$
 g := n+1 - d - \log _{|G|} (|C|).
 $$
\end{lemmadefinition}

\begin{proof}

Let us put $q := |G|$ and note that the subgroup $(G':= G^{d-1} \times 0_G ^{n-d+1}, +) < (G^n, +)$ has trivial intersection $C \cap G'= \{ 0_G ^n \}$ with $C$.
If $G'' := G'+ C$ is the subgroup  of $(G^n, +)$, generated by $G'$ and $C$ then any element of the quotient group $G'' /G'$ admits a representative from $0_G ^{d-1} \times G^{n-d+1}$.
In particular, $G'' / G'$ is of order $\left| G'' / G'\right| \leq \left| 0_G ^{d-1} \times G^{n-d+1} \right| = q^{n-d+1}$.
The natural projection
$$
\varphi : (C, +) \longrightarrow (G'' / G', +), \quad \varphi (c) := c + G'\ \ \mbox{  for }  \ \ \forall c \in C
$$
is a group  homomorphism with kernel $\ker \varphi = C \cap G'= \{ 0_G ^n \}$.
Therefore $\varphi$ is injective and
$$
\left|C \right| = \left| \varphi (C) \right| \leq \left| G'' / G' \right|  \leq q^{n-d+1}.
$$
The real logarithmic function with base $q = |G| >1$ is increasing, so that $\log _q (|C|) \leq n-d+1$ and $g \geq 0$.

\end{proof}

Among the additive codes $(C, +) \leq (G^n, +)$ over $G$ of length  $n$ and cardinality $k = \log _{|G|} (|C|)$, the code $C_o$ of genus $g=0$ has unique decoding up to maximal possible number  $\left[ \frac{n+1-k}{2} \right]$ of perturbed symbols,  we say that $C_o$ is maximum distance separable or an additive MDS-code and denote it by ${\rm MDS} (n, n+1-k)$.
Here is another trivial result, which will be used in the sequel.

\begin{lemma}    \label{MDSDualIdMDS}
If $(C, +) \lneq (G^n, +)$ is a non-trivial additive MDS-code of cardinality $|C| = |G|^k$  for some  $k \in \RR$, $k < n$ then the dual code
 $C^{\perp} := \{ \pi \in \widehat{G} ^n \, \vert \, \pi (c) = 1, \forall c \in C   \} \simeq \widehat{G^n /C}$ of  cardinality
  $\left| C^{\perp} \right| = \frac{|G|^n}{|C|} = |G|^{n-k}$ is ${\rm MDS} (n, k+1)$.
\end{lemma}

\begin{proof}

Note that an arbitrary character $\pi \in C^{\perp} \subset \widehat{G}^n$ provides a correctly defined homomorphism
$$
\overline{\pi} : (G^n/C, +) \longrightarrow  (S^1, .), \quad \overline{\pi} ( a+ C) := \pi (a) \ \ \mbox{  for  } \ \ \forall a \in G^n,
$$
according to $\pi (a +c) = \pi (a) \pi (c) = \pi (a)$ for $\forall c \in C \leq \ker \pi$.
Conversely,  any character $\overline{\pi} : (G^n / C, +) \rightarrow (S^1, .)$ of the quotient group $(G^n/ C, +) = (G^n, +) / (C, +)$ lifts to a character $\pi : (G^n, +) \rightarrow (S^1, .)$, $\pi (a) := \overline{\pi} (a + C)$ for $\forall a \in G^n$ with $C \leq \ker \pi$.
Therefore $\pi \in C^{\perp}$ and there is a group isomorphism $(C^{\perp}, .) \simeq ( \widehat{G^n /C}, .)$.
In particular, the cardinality $|C^{\perp}| = \left| G^n / C \right| = [ G^n : C] = \frac{\left| G^n \right|}{\left| C \right|} = |G|^{n-k}$.

The assumption $C \neq G^n$ implies that $C^{\perp} \neq \{ \varepsilon \}$ has minimum distance $d^{\perp} \in \NN$.
Note that  $k := \log _{|G|} (|C|) \in \RR$ is a real number,  $k < n$ and assume that the genus
$g^{\perp} := n+1 - d^{\perp} - \log _{|G|} (|C^{\perp}|) = n+1 - d^{\perp} - (n-k) = k+1 - d^{\perp} >0$ of $C^{\perp}$ is strictly  positive,
Then $d^{\perp} \leq k$ and for any $\pi \in C^{\perp}$ of weight  ${\rm wt} ( \pi) = d^{\perp}$ there exists a $k$-tuple of indices $\alpha  \in \binom{[n]}{k}$ with ${\rm Supp} ( \pi) \subseteq \alpha$.
Let $\beta = \neg \alpha := \{ 1, \ldots , n \} \setminus \alpha$ be the complement of $\alpha$ and
$$
\Pi _{\beta} : C \longrightarrow \Pi _{\beta} (C) \subseteq G^k
$$
 be the puncturing at $\beta$.
Note that $\Pi _{\beta}$ is a homomorphism of additive groups with $\ker ( \Pi _{\beta}) \cap C \neq \{ 0_G ^n \}$.
Otherwise, the restriction of $\Pi _{\beta}$ on $C$ is injective and $|G|^k = |C| = |\Pi _{\beta} (C)| \leq |G|^k$ implies that $\Pi _{\beta} (C) = G^k$.
However, $\pi$ has trivial components $\varepsilon$, labeled by $\beta$ and for $\forall c \in C$ there holds
$$
1 = \pi (c) = \prod\limits _{i=1} ^n \pi _i ( c_i) = \prod\limits _{i \in \alpha} \pi _i (c_i) = \Pi _{\beta} ( \pi) ( \Pi _{\beta} (c)).
$$
Thus, $\Pi _{\beta} ( \pi) \in \Pi _{\beta} (C)^{\perp} = (G^k) ^{\perp} = \{ \varepsilon ^k \}$, whereas $\pi = \varepsilon ^n$, contrary to the choice of $\pi \in C^{\perp}$ with ${\rm wt} ( \pi) = d^{\perp} \in \NN$.
That justifies $\ker \Pi _{\beta} \cap C \neq 0_G ^n$.
However, any word  $c \in (\ker \Pi _{\beta} \cap C) \setminus \{ 0_G ^n \}$ has support ${\rm Supp} (c) \subseteq \beta$ and, therefore, is of weight
 $1 \leq {\rm wt} (c) \leq  n-k$.
 By assumption, $C$ is of genus $g = n+1 -d-k =0$ or of minimum weight $d = n-k+1$ and does not contain non-zero words of weight $\leq n-k$.
 The contradiction justifies that $C^{\perp}$ is of genus $0$ or a  Maximum Distance Separable code ${\rm MDS} (n, k+1)$.

\end{proof}

In order to compute explicitly the weight distribution of an additive MDS-code $({\rm MDS} (n,d), +) < (G^n, +)$ of minimum distance $d$, one need one more lemma.

\begin{lemma}    \label{MDSWordMinWeight}
Let $(C, +) := {\rm MDS} (n,d) \neq \{ 0_G ^n \}$   be a non-zero additive MDS-code  of length $n$ and minimum distance $d \in \NN$, over a finite abelian group $G$ .
Then:

(i) in the case of $C \neq G^n$, for any $(d-1)$-tuple of indices $\beta \in \binom{[n]}{d-1}$ the puncturing $\Pi _{\beta} : (C, +) \rightarrow (G^{n+1-d}, +)$ is a group isomorphism onto $(G^{n+1-d}, +)$;

(ii) for any $\gamma \in \binom{[n]}{d}$ there are exactly $|G|-1$ words of $C$ with support $\gamma$;

(iii) $C$ has $\cM _{n,d} ^{(d)} = \binom{n}{d} (|G|-1)$ words of weight $d$.
\end{lemma}

\begin{proof}

(i) If $|C| < |G|^n$ then the additive MDS-code $C$ is of minimum distance $d = n+1 - \log _{|G|} (|C|) >1$ and the puncturing $\Pi _{\beta} : (G^n, +) \rightarrow (G^{n+1-d}, +)$ is non-trivial, i.e., non-identical.
Since $C$ is of minimum distance $d$, the kernel $\ker \Pi _{\beta} = \{ c \in G^n \, \vert \, {\rm Supp} (c) \subseteq \beta \}$ of $\Pi _{\beta}$ intersects $C$ at the origin $0_G^n$ alone  and $\Pi _{\beta} : C \rightarrow \Pi _{\beta} (C)$ is bijective.
Therefore $|\Pi _{\beta} (C)| = |C| = |G |^{n+1-d}= |G ^{n+1-d}|$ and $\Pi _{\beta} (C) = G^{n+1-d}$.
If $c, c'\in \Pi _{\beta} ^{-1} (a)$ for some $a \in G^{n+1 -d}$ then ${\rm Supp} (c - c') \subseteq \beta \in \binom{[n]}{d-1}$ and $c = c'$.
In such a way, the puncturing  $\Pi _{\beta} : (C, +) \rightarrow ( G^{n+1-d}, +)$ is shown to be  a group isomorphism.

(ii)  The additive code $C = G^n$ is of minimum distance $d=1$ and for any $\gamma \in \binom{[n]}{1}$ there are exactly $|G|-1$ words $(0_G^{\gamma -1}, g, 0_G^{n-\gamma}) \in  G^n$,  $g \in G \setminus \{ 0_G \}$ with support $\gamma$.
From now on, we assume that $C \lneq G^n$ is a proper subgroup of $(G^n, +)$.
If  $\gamma \in \binom{[n]}{d}$ and $i \in \gamma$ let $\beta := \gamma \setminus \{ i \} \in \binom{[n]}{d-1}$,
 $\delta := \neg \gamma = \{ 1, \ldots , n \} \setminus \gamma$ and recall from (i)  that the puncturing
 $\Pi _{\beta} : (C, +) \rightarrow (G^{n+1-d}, +)$ is a group isomorphism.
 In particular, for any $a \in G^{n+1-d}$ with $a_{\delta} = 0_G ^{n-d}$ and an arbitrary $a_i \in G \setminus \{ 0_G \}$ there is a unique $c \in C$ with $\Pi _{\beta} (c) =a$.
 Therefore, the support ${\rm Supp} (c) \subseteq \beta \cup \{ i \} = \gamma$ of $c$ is contained in $\gamma$ and since $C$ is of minimum distance $d$, there follows ${\rm Supp} (c) = \gamma$.
 In such a way, we have shown the existence of at least $|G|-1$ words of $C$ with support $\gamma$.
 Since any $c \in C$ with ${\rm Supp} (c) = \gamma$ is among the constructed ones, there are exactly $|G|-1$ words $c \in C$ with ${\rm Supp} (c) = \gamma$.

 (iii) is an immediate consequence of (ii) and the fact that  the number of the $d$-tuples $\gamma \in \binom{[n]}{d}$ is  $\binom{n}{d}$.

\end{proof}

The next proposition computes   the homogeneous weight enumerator $\cM_{n,d} (x,y)$ of an additive MDS-code ${\rm MDS} (n,d)$ of length $n$ and minimum distance $d$ over an arbitrary finite abelian group $G$.
That allows to express the homogeneous weight enumerator of an arbitrary additive code $(C, +) \leq (G^n, +)$ of minimum distance $d$ with dual $(C^{\perp}, .) \leq ( \widehat{G}^n, .)$ of  minimum distance $d^{\perp}$ as a $\QQ$-linear combination of the polynomials $\cM _{n,d} (x,y)$, $\cM _{n,d+1} (x,y)$, $\ldots$,
$\cM _{n, n+2- d^{\perp}}(x,y)$ (cf.Proposition-Definition \ref{ExistsUniqueDzetaPolynomial} from the next section).

\begin{proposition}   \label{MDSEnumerator}
An arbitrary additive MDS-code $({\rm MDS} (n,d), +) \leq (G^n, +)$ of minimum distance $d \in \NN$ has
\begin{equation}   \label{MDSWeightDistribution}
\cM _{n,d} ^{(s)} = \binom{n}{s} (|G|-1) \left[ \sum\limits _{i=0} ^{s-d} (-1)^{i} \binom{s-1}{i}  |G| ^{s-d-i} \right]
\end{equation}
words of weight $s$ for all  $  d \leq s \leq n$.
\end{proposition}

\begin{proof}

Let $q := |G|$ be the order of $G$.
If ${\rm MDS} (n,d) = G^n$ then $d=1$ and for any $1 \leq s \leq n$ there are $\binom{n}{s}$ subsets $\gamma \subseteq \{ 1, \ldots , n \}$ of cardinality $|\gamma| =s$.
  For   any such $\gamma$ there are $(q-1)^s$ words $( a_{\gamma}, a_{\neg \gamma} = 0_G ^{n-s}) \in G^n$, $a_{\gamma _1}, \ldots , a_{\gamma _s} \in G \setminus \{ 0_G \}$ with support $\gamma$.
Therefore $\cM _{n,1} ^{(s)} = \binom{n}{s} (q-1)^s$ for all $1 \leq s \leq n$.
Bearing in mind that
\begin{align*}
\binom{n}{s} (q-1) \left[ \sum\limits _{i=0} ^{s-d} (-1)^{i} \binom{s-1}{i} q^{s-d-i} \right] =
\binom{n}{s} (q-1) \left[ \sum\limits _{i=0} ^{s-1} (-1)^{i} \binom{s-1}{i} q^{s-1-i} \right] = \\
\binom{n}{s} (q-1) (q-1)^{s-1} = \binom{n}{s} (q-1)^s,
\end{align*}
 one proves (\ref{MDSWeightDistribution}) for ${\rm MDS} (n,d) = {\rm MDS} (n,1) = G^n$.

From now on, assume that the additive MDS-code  $({\rm MDS} (n,d), +) \lneq (G^n, +)$ is a  proper  subgroup of $(G^n, +)$  and note that the dual
 ${\rm MDS} (n,d)^{\perp}, .) \leq ( \widehat{G}^n, .)$ of cardinality $\left| {\rm MDS} (n,d)^{\perp} \right| = \frac{q^n}{\left| {\rm MDS} (n,d) \right|} = \frac{q^n}{q^{n+1-d}} = q^{d-1}$ is of genus $0$ by Lemma \ref{MDSDualIdMDS}.
Therefore $d^{\perp} = n+2 -d$ and
\begin{equation}    \label{DualOfMDS}
{\rm MDS} (n,d)^{\perp} = {\rm MDS} (n, n+2-d).
\end{equation}
In particular, $d^{\perp} \geq 2$ by $d \leq n$.
We note  that
$$
\mu _d ^{(s)} := (q-1) \left[ \sum\limits _{i=0} ^{s-d} (-1)^{i} \binom{s-1}{i} q^{s-d-i} \right] \ \ \mbox{  for  } \ \ \forall d \leq s \leq n
$$
is independent of $n$ and show that $\cM _{n,d} ^{(s)} = \binom{n}{s} \mu _d ^{(s)}$ by an induction on the length $n$.
If $n=1$ then the assumption $\left| {\rm MDS} (1,d) \right| = q^{2-d} < q = |G|$ requires $d>1$, which is an absurd.
Thus, the only additive MDS-code of length $1$ is ${\rm MDS} (1,1) = G$ and (\ref{MDSWeightDistribution}) is true for all additive MDS-codes of length $1$.
Assume that (\ref{MDSWeightDistribution})  holds  for all additive MDS-codes of length   $n-1$, put $C := {\rm MDS} (n,d)$  and consider the shortening $S_i : C \rightarrow S_i (C)$.
By its very definition, the shortening $S_i$ does not erase non-zero components and preserves the minimum distance  $d$.
On the other hand, the puncturing $\Pi _i : C^{\perp} \rightarrow \Pi _i (C^{\perp})$ of the  dual code $(C^{\perp}, .) \leq ( \widehat{G}^n, .)$ is bijective, as far as $\pi \in \ker \Pi _i \cap C^{\perp} $ exactly when ${\rm Supp} ( \pi ) \subseteq \{ i \}$ and  $d^{\perp} \geq 2$.
Making use of $\Pi _i (C^{\perp}) = S_i (C)^{\perp}$, established by (\ref{PuncturingAndShortening}) from  Lemma \ref{PuncturingAnd Shortening}, one concludes that
$$
\left|  S_i (C)^{\perp} \right| = \left| \Pi _i (C^{\perp}) \right| = \left| C^{\perp} \right| = q^{d-1}
$$
for $(S_i (C)^{\perp}, .) \leq (\widehat{G}^{n-1},.)$, whereas $\left| S_i (C) \right| = \frac{q^{n-1}}{\left| S_i (C)^{\perp} \right|} = q^{n-d}$.
Thus, $S_i (C)$ is of genus $g (S_i (C)) := (n-1)+1 - d( S_i (C)) - \log _q \left| S_i (C) \right|  =0$ and by the inductional hypothesis,  the homogeneous weight enumerator $\cW _{S_i (C)} (x,y) := x^{n-1} + \sum\limits  _{s=d} ^{n-1} \cW ^{(s)} _{S_i (C)} x^{n-1-s} y^s$ of $S_i (C)$ has coefficients $\cW_{S_i (C)} ^{(s)} = \binom{n-1}{s} \mu _d ^{(s)}$ for $\forall d \leq s \leq n-1$.
Let
$$
S := \sum\limits _{i=1} ^n S_i : \cM _{n,d} (x,y) \longrightarrow S \cM _{n,d} (x,y) = \sum\limits _{i=1} ^n S_i \cM _{n,d} (x,y) := \sum\limits _{i=1} ^n \cW _{S_i (C)} (x,y)
$$
be the map, which transforms  the homogeneous weight enumerator $\cM _{n,d} (x,y)$ of $C = {\rm MDS} (n,d)$ into the sum of the homogeneous weight enumerator of $S_i (C)$.
By the inductional hypothesis,
\begin{equation}   \label{WeightEnumeratorAverageShortening}
S \cM _{n,d} (x,y) = n x^{n-1} + \sum\limits _{s=d} ^{n-1} n \binom{n-1}{s} \mu _d ^{(s)} x^{n-1-s} y^s.
\end{equation}
On the other hand, an arbitrary monomial $x^{n-s} y^s$ of $\cM _{n,d} (x,y)$ counts a word $c \in C$ with support ${\rm Supp} (c) = \gamma = \{ \gamma _1, \ldots , \gamma _s \} \in \binom{[n]}{s}$.
The shortenings $S_{\gamma _j}$ with $1 \leq j \leq s$ do not produce a word from $\coprod_{i=1} ^n S_i (C)$, while the shortenings  $S_i$ at $i \in \{ 1, \ldots , n \} \setminus \gamma$  yield  a word of $S_i (C)$.
Thus, $x^{n-s}y^s$ is transformed into $(n-s) x^{n-1-s} y^s = \frac{\partial}{\partial x} (x^{n-s} y^s)$.
Since the partial derivative
$$
\frac{\partial}{\partial x} : \CC[x,y] ^{(n)} \longrightarrow \CC [x,y] ^{(n-1)}
$$
 is a $\CC$-linear map of the homogeneous polynomials $\CC [x,y] ^{(n)}$ of $x,y$ of degree $n$ in the homogeneous polynomials $\CC[x,y] ^{(n-1)}$ of degree $n-1$, the polynomial
 $$
 S \cM _{n,d} (x,y) = \frac{\partial}{\partial x} \cM _{n,d} (x,y).
 $$
Combining with (\ref{WeightEnumeratorAverageShortening}) and comparing the coefficients of $x^{n-1-s} y^s$ for $\forall d \leq s \leq n$, one concludes that
$(n-s) \cM _{n,d} ^{(s)} = n \binom{n-1}{s} \mu _d ^{(s)}$.
That completes the proof of
$$
\cM _{n,d} ^{(s)} = \frac{n \binom{n-1}{s}}{n-s} \mu _d ^{(s)} = \binom{n}{s} \mu _d ^{(s)} \ \ \mbox{  for  } \ \ \forall 1 \leq s \leq d.
$$

\end{proof}

The following corollary  will be useful for expressing Mac Williams identities for an arbitrary pair $C$, $C^{\perp}$ of mutually dual additive codes in terms of their $\zeta$-polynomials $P_C(t)$, $P_{C^{\perp}}(t) \in \QQ[t]$.
In order to formulate and prove it, we consider the set $\CC[x,y] ^{(n)} [[t]]$ of the formal power series of $t$, whose coefficients are homogeneous polynomials of $x,y$ of degree $n$.
For arbitrary $\eta (x,y,t) \in \CC[x,y] ^{(n)} [[t]]$ and $s \in \ZZ ^{\geq 0}$, let us  denote by ${\rm Coeff} _{t^s} ( \eta (x,y,t)) \in \CC [x,y] ^{(n)}$ the coefficient of $t^s$ from $\eta (x,y,t)$.
The proof of the proposition coincides with the one of Proposition 1 from Duursma's \cite{D2} on $\FF_q$-linear MDS-codes.

\begin{corollary}    \label{MDSEnumeratorAsCoefficients}
The homogeneous weight enumerator $\cM _{n,d} (x,y)$ of an additive MDS-code $({\rm MDS} (n,d), +) \leq (G^n, +)$ of minimum distance $d$ is uniquely determined by the equality of polynomials
\begin{equation}    \label{MDSEnumeratorExpressionAsCoefficient}
\frac{\cM _{n,d} (x,y) - x^n}{|G|-1} = {\rm Coeff} _{t^{n-d}} \left( \frac{[xt + y(1-t)] ^n}{(1-t)(1-|G|t)} \right).
\end{equation}

\end{corollary}

\begin{proof}

Let us denote $q := |G|$, $\eta (x,y,t) := \frac{[xt + y (1-t)]^n}{(1-t)(1-qt)}$ and note that
$$
\eta (x,y,t) = \sum\limits _{s=0} ^n \binom{n}{s} \frac{t^{n-s}(1-t)^s}{(1-t)(1-qt)} x^{n-s} y^s
$$
has coefficient
$$
{\rm Coeff} _{t^{n-d}} \eta (x,y,t) = \sum\limits _{s=0} ^n \binom{n}{s} {\rm Coeff} _{t^{s-d}} \left( \frac{(1-t)^{s-1}}{1-qt} \right) x^{n-s} y^s
$$
of $t^{n-d}$.
 Since  $\frac{(1-t)^{s-1}}{1-qt} \in \CC [[t]]$ has no pole at $t=0$,  one has  ${\rm Coeff} _{t^{s-d}} \left( \frac{(1-t)^{s-1}}{1-qt} \right) =0$ for $\forall 0 \leq s \leq d-1$ and
 $$
 {\rm Coeff} _{t^{n-d}} \eta (x,y,t) = \sum\limits _{s=d} ^n \binom{n}{s} {\rm Coeff} _{t^{s-d}} \left( \frac{(1-t)^{s-1}}{1-qt} \right) x^{n-s} y^s.
 $$
Making use of (\ref{MDSWeightDistribution}) from  Proposition \ref{MDSEnumerator}, one reduces the proof of (\ref{MDSEnumeratorExpressionAsCoefficient}) to
\begin{align*}
{\rm Coeff} _{t^{s-d}} \left( \frac{(1-t)^{s-1}}{1-qt} \right) =
{\rm Coeff} _{t^{s-d}} \left( \left[ \sum\limits _{i=0} ^{s-1} \binom{s-1}{i} (-1) ^{i} t^{i} \right] \left( \sum\limits _{j=0} ^{\infty} q^{j} t^{j} \right) \right) =  \\
\sum\limits _{i=0} ^{s-d} \binom{s-1}{i} (-1) ^{i} q^{s-d-i} \ \ \mbox{  for  }  \ \  \forall d \leq s \leq n.
\end{align*}

\end{proof}

%%%%%%%%%%%%%%%%%%%%%%%%%%%%%%%%%%%%%%%%%%%%%%%%%%%%%%%%%%%%%%%%%%%%%%%%%%%%%%%%%%%%%%%%%%%%%%%%%%%%%%%%%%%%%%%%%%%%%%%%%%%%%%%%%%%%%%%%%%%%%%%%%%%%%%%%%%%%%
\section{ Mac Williams  identities  as a functional equation  of \\ Duursma's reduced polynomials  }
 %%%%%%%%%%%%%%%%%%%%%%%%%%%%%%%%%%%%%%%%%%%%%%%%%%%%%%%%%%%%%%%%%%%%%%%%%%%%%%%%%%%%%%%%%%%%%%%%%%%%%%%%%%%%%%%%%%%%%%%%%%%%%%%%%%%%%%%%%%%%%%%%%%%%%%%%%%

The next proposition reminds Duursma's definition of a $\zeta$-polynomial $P_C(t)$ of an $\FF_q$linear code $C \subset \FF_q^n$  and expresses Mac Williams identities for the weight distribution of $C$, $C^{\perp} \subset \FF_q^n$ as a functional equation on $P_C(t)$, $P_{C^{\perp}}(t)$.
All properties of $C$, $C^{\perp}$, which are used by Duursma's construction hold for additive codes $(C,+) \leq (G^n, +)$ and their duals $(C^{\perp},.)\leq (\widehat{G}^n, .)$.
We formulate in terms of additive codes and provide the proofs for completeness.

\begin{propositiondefinition}    \label{ExistsUniqueDzetaPolynomial}
Let $(C, +) \leq (G^n, +)$ be an additive code of genus $g$ and minimum distance $d \geq 2$ with dual $(C^{\perp}, .) \leq ( \widehat{G}^n, .)$ of genus $g^{\perp}$ and minimum distance $d^{\perp} \geq 2$.
Then there exist unique polynomials
$$
P_C(t) = \sum\limits _{i=0} ^{g + g^{\perp}} a_i t^{i},  \ \ P_{C^{\perp}} (t) = \sum\limits _{i=0} ^{g + g^{\perp}} a_i ^{\perp} t^{i} \in \QQ[t]
 $$
 with $P_C(1) = P_{C^{\perp}} (1) =1$, whose coefficients express the homogeneous weight enumerators
\begin{equation}   \label{EnumeratorCByP}
\cW _C (x,y) = \sum\limits _{i=0} ^{g + g^{\perp}} a_i \cM _{n, d+i} (x,y),
\end{equation}
\begin{equation}    \label{EnumeratorCPerpByP}
\cW _{C^{\perp}} (x,y) = \sum\limits _{i=0} ^{g + g^{\perp}} a_i ^{\perp} \cM _{n, d^{\perp} +i} (x,y)
\end{equation}
by  appropriate  MDS weight enumerators.
Mac Williams identities for $C$, $C^{\perp}$ are equivalent to the functional equation
\begin{equation}   \label{MacWilliamsDzetaPolynomial}
P_{C^{\perp}} (t) = P_C \left( \frac{1}{|G|t} \right) |G|^g t^{g + g^{\perp}}
\end{equation}
for $P_C(t)$, $P_{C^{\perp}} (t)$.
Form now on, $P_C(t)$, $P_{C^{\perp}} (t)$ are referred to as the $\zeta$-polynomials of $C$, $C^{\perp}$ and the rational functions
\begin{equation}   \label{DzetaFunctionsOfCCPerp}
\zeta _C(t) := \frac{P_C(t)}{(1-t)(1 - |G|t)}, \ \  \zeta _{C^{\perp}} (t) := \frac{P_{C^{\perp}} (t)}{(1-t)(1 - |G|t)}
\end{equation}
are called  the $\zeta$-functions of $C$, $C^{\perp}$.
\end{propositiondefinition}

\begin{proof}

For arbitrary $n, d \in \NN$, $d \leq n$ let $\QQ [x,y] ^{(n)} _{\geq d}$ be the $\QQ$-linear space of the homogeneous polynomials of $x,y$ of degree $n$, whose monomials with non-zero coefficients are of degree $\geq d$ with respect to $y$.
Then $\{ x^{n-i} y^{i} \, \vert \, d \leq i \leq n \}$ is a $\QQ$-basis of $\QQ[x,y] ^{(n)} _{\geq d}$, as well as any set of polynomials $f_i (x,y) \in \QQ[x,y] ^{(n)} _{\geq i}$, $d \leq i \leq n$ with ${\rm Coeff} _{x^{n-i} y^{i}} ( f_i (x,y)) \neq 0$ for all $d \leq i \leq n$.
In particular, for $\forall 0 \leq i \leq n-d$ the polynomials
$$
\cM _{n, d+i} (x,y) - x^n = \sum\limits _{s = d+i} ^n \cM _{n, d+i} ^{(s)} x^{n-s} y^s
$$
with $\cM _{n, d+i} ^{(d+i)} = \binom{n}{d+i} (|G|-1) >0$ constitute a $\QQ$-basis of $\QQ[x,y] ^{(n)} _{\geq d}$.
The coordinates $a_0, \ldots , a_{n-d} \in \QQ$ of $\cW _C (x,y) - x^n \in \QQ [x,y] ^{(n)} _{\geq d}$ with respect to $\cM _{n, d} (x,y) - x^n$, $\ldots$, $\cM _{n,n} (x,y) - x^n$ provide a uniquely determined $\zeta$-polynomials $P_C(t) = \sum\limits _{i=0} ^{n-d} a_i t^{i} \in \QQ[t]$ with
\begin{equation}    \label{RoughtExpressionC}
\cW _C(x,y) - x^n = \sum\limits _{i=0} ^{n-d} a_i \left[ \cM _{n, d+i} (x,y) - x^n \right].
\end{equation}
Similar considerations provide a unique $\zeta$-polynomial $P_{C^{\perp}} (t) = \sum\limits _{i=0} ^{n - d^{\perp}} a_i ^{\perp} t^{i} \in \QQ[t]$ of $(C^{\perp}, .) \leq ( \widehat{G}^n, .)$ with
\begin{equation}    \label{RoughExpressionCPerp}
\cW _{C^{\perp}} (x,y) - x^n = \sum\limits _{i=0} ^{n - d^{\perp}} a_i ^{\perp} \left[ \cM _{n, d^{\perp} +i} (x,y) - x^n \right].
\end{equation}
According to Delsarte's \cite{Delsarte} (cf.also Wood's \cite{Wood}), Mac Williams identities for the weight distribution of an additive code $(C, +) \leq (G^n, +)$ and its dual $(C^{\perp}, .) \leq ( \widehat{G}^n, .)$ reads as
\begin{equation}   \label{GeneralMacWilliams}
|C| \cW _{C^{\perp}} (x,y) = \cW _C (x + (|G|-1) y, x-y).
\end{equation}
By assumption $d \geq 2$, so that $({\rm MDS} (n,d), +) \lneq (G^n, +)$ is a non-trivial MDS-code and its dual ${\rm MDS} (n,d)^{\perp} = {\rm MDS} (n, n+2-d)$ is also MDS by  (\ref{DualOfMDS}).
Denoting $q := |G|$, one expresses  Mac Williams identities (\ref{GeneralMacWilliams}) for ${\rm MDS} (n,d)$ and ${\rm MDS} (n,d)^{\perp}$   as
\begin{equation}   \label{MDSMacWilliams}
q^{n+1-d} \cM _{n, n+2-d} (x,y) = \cM _{n,d} ( x + (q-1)y, x-y).
\end{equation}
If $k := \log _q (|C|)$ then  substituting   (\ref{RoughtExpressionC}), (\ref{RoughExpressionCPerp}) in (\ref{GeneralMacWilliams}) and using (\ref{MDSMacWilliams}), one rewrites  (\ref{GeneralMacWilliams}) in the form
\begin{align*}
q^k \left \{ \sum\limits _{i=0} ^{ n - d^{\perp}} a_i ^{\perp} \left[ \cM _{n, d^{\perp} +i} (x,y) - x^n \right]  + x^n\right \}  =  \\
q^{k} [ \cW _{C^{\perp}} (x,y) - x^n] + q^{k} x^n = \\
\left \{ \cW _C (x + (q-1)y, x-y) - [x + (q-1) y] ^n \right \} + [x + (q-1)y] ^n =  \\
\sum\limits _{i=0} ^{n-d} a_i \left \{ \cM _{n, d+i} (x + (q-1) y, x-y) - [x + (q-1)y]^n \right \} + [x + (q-1)y] ^n =  \\
\sum\limits _{i=0} ^{n-d} a_i \left \{ q^{n+1 -d-i} \cM _{n, n+2 -d-i} (x,y) - [x + (q-1)y] ^n \right \} + [x + (q-1)y] ^n =  \\
\sum\limits _{i=0} ^{n-d} a_i q^{n+1-d-i} [ \cM _{n, n+2-d-i} (x,y) - x^n] +  \\
 + \left( \sum\limits _{i=0} ^{n-d} a_i q^{n+1 -d-i} \right) x^n +
\left(1 - \sum\limits _{i=0} ^{n-d} a_i \right) [x + (q-1) y] ^n =  \\
\sum\limits _{i=0} ^{n-d} a_i q^{n+1-d-i} [ \cM _{n, n+2-d-i} (x,y) - x^n] +  \\
+ q^{n+1-d} P_C \left( \frac{1}{q} \right) x^n +
\left[1 - P_C(1) \right] [x + (q-1)y] ^n.
\end{align*}
In terms of the summation indices, which equal the minimum distances of the corresponding MDS  weight enumerators, the above reads as
\begin{equation}   \label{IntermediateGeneralMacWilliams}
\begin{split}
q^k \left \{ \sum\limits _{j = d^{\perp}} ^n a^{\perp} _{j - d^{\perp}} [ \cM _{n,j} (x,y) - x^n ] \right \} + q^k x^n = \\
\sum\limits _{j=2} ^{n+2-d} a_{n+2-d-j} q^{j-1} [ \cM _{n, j} (x,y) - x^n ] + q^{n+1-d} P_C \left( \frac{1}{q} \right) x^n +
\left[ 1 - P_C(1) \right] [x + (q-1)y] ^n.
\end{split}
\end{equation}
Bearing in mind that $d^{\perp} \geq 2$, one compares the coefficients of $x^{n-1}y$ from (\ref{IntermediateGeneralMacWilliams}) and concludes that $P_C(1) =1$.
As a result, (\ref{RoughtExpressionC}) reduces to (\ref{EnumeratorCByP}) and (\ref{RoughExpressionCPerp}) takes the form (\ref{EnumeratorCPerpByP}), due to the presence of a canonical isomorphism $(C^{\perp}) ^{\perp} \simeq C$.
Then the comparison of the coefficients of $x^n$ yields $q^k = q^{n+1-d} P_C \left( \frac{1}{q} \right)$, whereas $P_C \left( \frac{1}{q} \right) = q^{ -g} = \left( \frac{1}{q} \right) ^g$ for the genus $g:= n+1 -d-k$ of $C$.
The comparison of the coefficients of $\cM _{n, j} (x,y) - x^n$ with $j < d^{\perp}$ or $j > n+2-d$ from (\ref{IntermediateGeneralMacWilliams}) implies $a_i = a_i ^{\perp} =0$ for all $i > n+2-d-d^{\perp} = g + g^{\perp}$ and allows to write Mac Williams identities (\ref{IntermediateGeneralMacWilliams}) for $C$, $C^{\perp}$ in the form
$$
q^k \sum\limits _{j = d^{\perp}} ^{n+2-d} a^{\perp} _{j - d^{\perp}} [ \cM _{n, j} (x,y) - x^n ] =
\sum\limits _{j = d^{\perp}} ^{n+2-d} q^{j-1} a_{n+2-d-j} [ \cM _{n, j} (x,y) - x^n ],
$$
which is equivalent to
$$
q^k a^{\perp} _{j - d^{\perp}} = q^{j-1} a_{n+2-d-j} \ \ \mbox{  for } \ \ \forall d^{\perp} \leq j \leq n+2-d.
$$
Under the substitution $i := j - d^{\perp}$, these amount to
\begin{equation}   \label{MacWilliamsForDzetaCoefficients}
a_i ^{\perp} = q^{i - g^{\perp}} a_{g + g^{\perp} -i} \ \ \mbox{  for  } \ \ \forall 0 \leq i \leq g + g^{\perp}.
\end{equation}
 If $P_C(t) := \sum\limits _{i=0} ^{g + g^{\perp}} a_i t^{i}, \ \ P_{C^{\perp}} (t) := \sum\limits _{i=0} ^{g + g^{\perp}} a_i ^{\perp} t^{i} \in \QQ[t]$
 then one easily checks that the equalities of the coefficients of $t^{i}$, $0 \leq i \leq g + g^{\perp}$ from both sides of (\ref{MacWilliamsDzetaPolynomial}) can be expressed as (\ref{MacWilliamsForDzetaCoefficients}).

 \end{proof}

 \begin{corollary}   \label{CoefficientDefDzetaPolynomial}
 If $(C, +) \leq (G^n, +)$ is an additive code of minimum distance $d \geq 2$ with dual $(C^{\perp}, .) \leq ( \widehat{G}^n, .)$ of minimum distance $d^{\perp} \geq 2$  then the $\zeta$-polynomials of $C$, $C^{\perp}$ are the  uniquely determined polynomials $P_C(t) = \sum\limits _{i=0} ^{g + g^{\perp}} a_i t^{i} \in \QQ[t]$, respectively, $P_{C^{\perp}} (t) = \sum\limits _{i=0} ^{g + g^{\perp}} a_i ^{\perp} t^{i} \in \QQ[t]$ of degree $\deg P_C(t) \leq g + g^{\perp}$,
 $\deg P_{C^{\perp}} (t) \leq g + g^{\perp}$ with $P_C(1) = P_{C^{\perp}} (1) =1$, which satisfy the equalities
 \begin{equation}   \label{CoefficientsDefinitionOfPC}
 \frac{\cW _C(x,y) - x^n}{|G|-1} = {\rm Coeff} _{t^{n-d}} \left( P_C(t) \frac{ [xt + y(1-t)]^n}{(1-t)(1-|G|t)} \right),
 \end{equation}
 respectively,
 \begin{equation}   \label{CoefficientsDefinitionOfPCPerp}
 \frac{\cW_{C^{\perp}} (x,y) - x^n}{|G|-1} = {\rm Coeff} _{t^{n - d^{\perp}}} \left( P_{C^{\perp}} (t) \frac{ [xt + y (1-t)]^n}{(1-t)(1-|G|t)} \right)
 \end{equation}
 of homogeneous polynomials of $x,y$ of degree $n$.
\end{corollary}

 \begin{proof}

Let $q := |G|$ be the order of $G$.
For an arbitrary polynomial $P_C(t) = \sum\limits _{i \geq 0} a_i t^{i} \in \QQ[t]$, whose degree will be bounded in the sequel, note that
\begin{align*}
{\rm Coeff} _{t^{n-d}} \left( P_C(t) \frac{[xt + y (1-t)]^n}{(1-t)(1-qt)} \right) =
{\rm Coeff} _{t^{n-d}} \left( \left[ \sum\limits _{i \geq 0} a_i t^{i} \right] \frac{[xt + y (1-t)]^n}{(1-t)(1-qt)} \right) =  \\
\sum\limits _{i \geq 0} a_i {\rm Coeff} _{t^{n-d-i}} \left( \frac{[xt + y (1-t)] ^n}{(1-t)(1-qt)} \right).
\end{align*}
The rational function $\frac{[xt + y (1-t)]^n}{(1-t)(1-qt)}$ without a pole at $t=0$ has vanishing coefficients of $t^{n-d-i}$ for $\forall i > n-d$, so that
the polynomial   $P_C(t) = \sum\limits _{i=0} ^{n-d} a_i t^{i} \in \QQ[t]$ is  of degree $\deg P_C(t) \leq n-d$.
Now, (\ref{MDSEnumeratorExpressionAsCoefficient}) and (\ref{RoughtExpressionC}) specify that

\begin{align*}
{\rm Coeff} _{t^{n-d}} \left( P_C(t) \frac{[xt + y (1-t)]^n}{(1-1)(1-qt)}  \right) =
\sum\limits _{i=0} ^{n-d} a_i \left( \frac{\cM _{n, d+i} (x,y) - x^n}{q-1}  \right)   =
\frac{\cW_C (x,y) - x^n}{q-1}.
\end{align*}
Similar considerations provide (\ref{CoefficientsDefinitionOfPCPerp}) for $P_{C^{\perp}} (t) = \sum\limits _{i=0}  ^{n - d^{\perp}} a_i ^{\perp} t^{i} \in \QQ[t]$.
The remaining part of the proof of Proposition-Definition \ref{ExistsUniqueDzetaPolynomial}  establishes that $\deg P_C(t) \leq g + g^{\perp}$, $\deg P_{C^{\perp}} (t) \leq g + g^{\perp}$ and $P_C(1) = P_{C^{\perp}} (1) =1$.

 \end{proof}

In order to get an impression of the (integral) denominators of the coefficients $a_i \in \QQ$ of $P_C(t)$, let us plug in $\cM _{n, d+i} (x,y) = x^n + \sum\limits _{s = d+i} ^n \cM _{n, d+i} ^{(s)} x^{n-s} y^s$ in (\ref{EnumeratorCByP}) and exchange the summation order.
That provides
\begin{align*}
x^n + \sum\limits _{s=d} ^n \cW _C ^{(s)} x^{n-s} y^s =  \\
x^n + \sum\limits _{s=d} ^{n +2 - d^{\perp}} \left( \sum\limits _{i=0} ^{s-d} a_i \cM _{n, d+i} ^{(s)} \right) x^{n-s} y^s +
\sum\limits _{s = n +3 - d^{\perp}} ^n \left( \sum\limits _{i=0} ^{g + g^{\perp}} a_i \cM _{n, d+i} ^{(s)} \right) x^{n-s} y^s,
\end{align*}
due to $P_C(1) =1$.
Comparing the coefficients of $x^{n-s}y^s$ for $\forall d \leq s \leq n+2-d^{\perp}$, one concludes that  $a_i$,
 $\forall 0 \leq i \leq n+2 - d^{\perp} -d = g + g^{\perp}$  constitute  a solution of the linear system of equations
$$
\left(   \begin{array}{ccccc}
\cM _{n,d} ^{(d)}      &  \ldots  &  0  &  \ldots  &  0  \\
\cM _{n,d}  ^{(d+1)}     &  \ldots  &  0  &  \ldots  &  0  \\
\ldots     & \ldots  & \ldots  & \ldots  & \ldots  \\
\cM_{n,d} ^{(s)}     &  \ldots  &  \cM _{n,s} ^{(s)}  &  \ldots  &  0  \\
\ldots     & \ldots  & \ldots  & \ldots  & \ldots  \\
\cM _{n,d} ^{(n+2-d^{\perp})}    &  \ldots  &  \cM _{n,s} ^{( n+2-d^{\perp})}  &  \ldots   &  \cM _{n, n+2-d^{\perp}} ^{(n+2-d^{\perp})}
\end{array}   \right)
\left(   \begin{array}{c}
a_0  \\
a_1  \\
\ldots  \\
a_{s-d}  \\
\ldots  \\
a_{g + g^{\perp}}
\end{array}  \right) =
\left(  \begin{array}{c}
\cW_C ^{(d)}  \\
\cW_C^{(d+1)}  \\
\ldots  \\
\cW _C ^{(s)}  \\
\ldots  \\
\cW _C ^{(g + g^{\perp})}
\end{array}  \right)
$$
with lower triangular coefficient matrix and  determinant
$$
\Delta = \prod\limits _{s=d} ^{n+2-d^{\perp}} \cM _{n,s} ^{(s)} = (q-1) ^{n+3-d^{\perp} -d} \prod\limits _{s=d} ^{n+2-d^{\perp}} \binom{n}{s} =
(q-1) ^{g + g^{\perp} + 1} \prod\limits _{j=0} ^{g + g^{\perp}} \binom{n}{d+j} \neq 0
$$
by (\ref{MDSWeightDistribution}).
  Cramer's rule applies to provide  $a_i  = \frac{\Delta _i}{\Delta}$  with  $\Delta _i \in \ZZ$  for  $\forall 0 \leq i \leq g + g^{\perp}$, so that
$$
\Delta a_i = (q-1) ^{g + g^{\perp} + 1} \prod\limits _{j=0} ^{g + g^{\perp}} \binom{n}{d+j} a_i \in \ZZ.
$$
In other words, the   denominators of $a_i$ are integral divisors of   $(q-1) ^{g + g^{\perp} + 1} \prod\limits _{j=0} ^{g + g^{\perp}} \binom{n}{d+j}$ for $\forall 0 \leq i \leq g + g^{\perp}$.

In the proof of Proposition-Definition \ref{ExistsUniqueDzetaPolynomial} we have established that the $\zeta$-polynomial
$P_C(t) = \sum\limits _{i=0} ^{g + g^{\perp}} a_i t^{i} \in \QQ[t]$ of an additive code $(C, +) \leq (G^n, +)$ of genus $g$ with dual
 $(C^{\perp}, .) \leq ( \widehat{G}^n, .)$ of genus $g^{\perp}$ has values $P_C(1) =1$ and $P_C \left( \frac{1}{|G|} \right) = \left( \frac{1}{|G|} \right) ^g$ at the simple poles of the $\zeta$-function $\zeta _C (t) = \frac{P_C(t)}{(1-t)(1-|G|t)}$, whenever $d \geq 2$ and $d^{\perp} \geq 2$.
Therefore $P_C(t) - t^g  \in \QQ[t]$ vanishes at $t=1$, $t = \frac{1}{|G|}$ and
$$
D_C(t) := \frac{P_C(t) - t^g}{(1-t)(1-qt)} \in \QQ [t]
$$
is a polynomial of degree $\deg D_C (t) \leq g + g^{\perp} -2$.
If  $g = g^{\perp} =0$, then $P_C(t) \equiv P_{C^{\perp}}(t) \equiv 1$ are constant and the  polynomials $D_C(t) \equiv D_{C^{\perp}}(t) \equiv  0$ vanish  identically.
In the case of an $\FF_q$-linear code $C$, \cite{KM1} refers to $D_C(t)$ as to Duursma's reduced polynomial of $C$ and making use of (\ref{EnumeratorCByP}) expresses the weight distribution of $C$ by the means of the coefficients $c_i$  of  $D_C(t) = \sum\limits _{i=0} ^{g + g^{\perp}-2} c_i t^{i} \in \QQ[t]$.
The present note adopts another definition of $D_C(t)$ and by the means of (\ref{CoefficientsDefinitionOfPC}) establishes its equivalence to the aforementioned one.
Duursma's reduced polynomials $D_C(t)$, $D_{C^{\perp}}(t)  \in \QQ[t]$ are used for expressing Mac Williams identities for $(C, +) < (G^n, +)$, $(C^{\perp}, .) < (\widehat{G}^n, .)$ as Polarized Riemann-Roch Conditions.
Besides, (\ref{DuursmasCoefficientsC})  from the next Proposition \ref{WCByDCAnd CoeffDenominators} reveals that
 the denominators of the  coefficients $c_i$  of   $D_C(t)$ are integral divisors of
  $(|G|-1) \binom{n}{d+i}$ for $\forall 0 \leq i \leq g + g^{\perp}-2$.
   In the case of $\FF_q$-linear codes $C, C^{\perp} \subset \FF_q ^n$,  Corollary \ref{DenominatorsCoeffDCLinear} specifies that  the  denominators  of $c_i$ divide $\binom{n}{d+i}$ for $\forall 0 \leq i \leq g + g^{\perp} -2$.

\begin{proposition}    \label{WCByDCAnd CoeffDenominators}
{\rm (Compare with Proposition 1 from \cite{KM1})} Let $(C, +) \leq (G^n, +)$ be an additive code of cardinality $|G|^k$ and minimum distance $d \geq 2$ with dual $(C^{\perp},.) \leq (\widehat{G}^n, .)$ of cardinality $|G|^{n-k}$ and minimum distance $d^{\perp} \geq 2$.
Then there exist unique Duursma's reduced polynomials $D_C(t) = \sum\limits _{i=0} ^{g + g^{\perp} -2} c_i t^{i} \in \QQ[t]$, respectively, $D_{C^{\perp}} (t) = \sum\limits _{i=0} ^{g + g^{\perp} -2} c_i ^{\perp} t^{i} \in \QQ[t]$ of $\deg D_C(t) \leq g + g^{\perp} -2$, $\deg D_{C^{\perp}} (t) \leq g + g^{\perp} -2$, such that
\begin{equation}    \label{EnumeraotrExpressionCByDRP}
\cW_C (x,y) = \cM _{n, n+1-k} (x,y) + \sum\limits _{i=0} ^{g + g^{\perp} -2} \left[ \binom{n}{d+i} (|G|-1) c_i \right] (x-y) ^{n-d-i} y^{d+i},
\end{equation}
respectively,
\begin{equation}    \label{EnumeratorExpressionCPerpByDRP}
\cW_{C^{\perp}} (x,y) = \cM _{n, k+1} (x,y) + \sum\limits _{i=0} ^{g + g^{\perp} -2} \left[ \binom{n}{d^{\perp}+i} (|G|-1) c_i ^{\perp} \right] (x-y) ^{n - d^{\perp} -i} y ^{d^{\perp} +i}
\end{equation}
for the genus $g = n+1 -d-k$ of $C$ and the genus $g^{\perp} = k+1 -d^{\perp}$ of $C^{\perp}$.

The polynomials $D_C(t)$, $D_{C^{\perp}} (t)$ can be determined by the equalities
\begin{equation}   \label{DuursmaByDzetaPolynomials}
D_C(t) = \frac{P_C(t) - t^g}{(1-t)(1-|G|t)}, \ \ \mbox{\rm  respectively,  } \ \ D_{C^{\perp}} (t) = \frac{P_{C^{\perp}} (t) - t^{g^{\perp}}}{(1-t)(1-|G|t)}.
\end{equation}

 Mac Williams identities for the weight distributions of $C$, $C^{\perp}$ are equivalent to the functional equation
\begin{equation}   \label{MacWilliamsForDRP}
D_{C^{\perp}} (t) = D_C \left( \frac{1}{|G|t} \right) |G|^{g-1} t^{g + g^{\perp} -2}
\end{equation}
 of  the corresponding Duursma's reduced polynomials and
\begin{equation}  \label{DuursmasCoefficientsC}
(|G|-1) \binom{n}{d+i} c_i = \sum\limits _{s=0} ^{d+i} \left(  \cW_C ^{(s)} - \cM _{n, n+1-k} ^{(s)} \right) \binom{n-s}{n-d-i} \in \ZZ,
\end{equation}
\begin{equation}   \label{DuursmasCoefficientsCPerp}
(|G|-1) \binom{n}{d^{\perp} +i} c_i ^{\perp} = \sum\limits _{s=0} ^{ d^{\perp} +i} \left(  \cW_{C^{\perp}} ^{(s)} - \cM _{n, k+1} ^{(s)} \right)  \binom{n-s}{n-d^{\perp}-i} \in \ZZ
\end{equation}
for $\forall 0 \leq i \leq g + g^{\perp} -2.$.
\end{proposition}

\begin{proof}

Towards the existence of $D_C(t) \in \QQ[t]$, let $q := |G|$ be the order of $G$ and
$$
\widetilde{D_C} (t) := \frac{P_C(t) - t^{g}}{(1-t)(1-qt)} = \sum\limits _{i=0} ^{g + g^{\perp} -2} c_i t^{i} \in \QQ[t].
$$
Then $P_C(t) = (1-t)(1-qt) \left( \sum\limits _{i=0} ^{g + g^{\perp} -2} c_i t^{i} \right) + t^g$ and (\ref{CoefficientsDefinitionOfPC}) can be written in the form
\begin{align*}
\frac{\cW_C (x,y) - x^n}{q-1} =  \\
 \sum\limits _{i=0} ^{g + g^{\perp} -2} c_i {\rm Coeff} _{t^{n-d-i}} ( [xt + y (1-t)] ^n) +
{\rm Coeff} _{t^{n-d-g}} \left( \frac{[xt + y (1-t)] ^n}{(1-t)(1-qt)} \right) =  \\
\sum\limits _{i=0} ^{g + g^{\perp} -2} c_i \left[ \sum\limits _{s=0} ^n \binom{n}{s} x^{n-s} y^s {\rm Coeff} _{t^{s-d-i}} (1-t)^s \right] +
\frac{\cM _{n, n-k+1} (x,y) - x^n}{q-1},
\end{align*}
 making use of  (\ref{MDSEnumeratorExpressionAsCoefficient}).
Note that ${\rm Coeff} _{t^{s-d-i}} (1-t)^s=0$ for $\forall s < d+i$, because $(1-t)^s$ has no pole at $t=0$ and $\binom{n}{s} \binom{s}{s-d-i} =
\binom{n}{d+i} \binom{n-d-i}{s-d-i}$.
As a result,
\begin{align*}
\cW_C(x,y) - \cM_{n,n+1-k} (x,y) =  \\
  \sum\limits _{i=0} ^{g + g^{\perp}-2} \sum\limits _{s=d+i} ^n (q-1) \binom{n}{s} \binom{s}{s-d-i} (-1) ^{s-d-i} c_i x^{n-s} y^s = \\
\sum\limits _{i=0} ^{g + g^{\perp} -2} \left[ \sum\limits _{s = d+i} ^n  \binom{n-d-i}{s-d-i} (-1)^{s-d-i} x^{n-s} y^{s-d-i} \right] (q-1) \binom{n}{d+i} c_i y^{d+i} =  \\
\sum\limits _{i=0} ^{g + g^{\perp} -2} \left[ \sum\limits _{j=0} ^{n-d-i} \binom{n-d-i}{j} (-1) ^{j} x^{n-d-i-j} y^{j} \right] \left[ (q-1) \binom{n}{d+i} c_i y^{d+i} \right] = \\
\sum\limits _{i=0} ^{g + g^{\perp} -2} \left[ (q-1) \binom{n}{d+i} c_i \right] (x-y) ^{n-d-i} y^{d+i}.
\end{align*}
Similar considerations for
$$
\widetilde{D_{C^{\perp}}} (t) := \frac{P_{C^{\perp}} (t) - t^{g^{\perp}}}{(1-t)(1-qt)} = \sum\limits _{i=0} ^{g + g^{\perp}-2} c_i ^{\perp} t^{i}
$$
provide (\ref{EnumeratorExpressionCPerpByDRP}).
This shows the existence of polynomials  $D_C (t) = \sum\limits _{i=0} ^{g + g^{\perp} -2} c_i t^{i}$, $D_{C^{\perp}} (t) = \sum\limits _{i=0} ^{ g + g^{\perp} -2} c_i ^{\perp} t^{i} \in \QQ[t]$, subject to (\ref{EnumeraotrExpressionCByDRP}), (\ref{EnumeratorExpressionCPerpByDRP}).

Towards the uniqueness of the polynomials $D_C(t)$, $D_{C^{\perp}}$, whose coefficients satisfy (\ref{EnumeraotrExpressionCByDRP}), (\ref{EnumeratorExpressionCPerpByDRP}), let us introduce $z := x-y$ and note that (\ref{EnumeraotrExpressionCByDRP}) is equivalent to
\begin{equation}   \label{EnumeratorCByYZ}
\sum\limits _{i=0} ^{g + g^{\perp} -2} \left[ \binom{n}{d+i} (q-1) c_i \right] y^{d+i} z^{n-d-i} = \cW_C (y+z, y) - \cM _{n, n+1-k} (y+z, y)
\end{equation}
and (\ref{EnumeratorExpressionCPerpByDRP}) amounts to
\begin{equation}   \label{EnumeratorCPerpByYZ}
\sum\limits _{i=0} ^{g + g^{\perp} -2} \left[ \binom{n}{d^{\perp} +i} (q-1) c_i ^{\perp} \right] y^{d^{\perp} +i} z^{n-d^{\perp} -i} =
\cW _{C^{\perp}} (y+z, y) - \cM _{n, k+1} (y+z, y).
\end{equation}
Therefore (\ref{EnumeraotrExpressionCByDRP}), respectively, (\ref{EnumeratorExpressionCPerpByDRP}) determine uniquely $c_i \in \QQ$, respectively,
$c_i ^{\perp} \in \QQ$ for $\forall 0 \leq i \leq g + g^{\perp} -2$ and Duursma's reduced polynomials $D_C(t) = \sum\limits _{i=0} ^{g + g^{\perp} -2} c_i t^{i}$, respectively, $D_{C^{\perp}} (t) = \sum\limits _{i=0} ^{g + g^{\perp} -2} c_i ^{\perp} t^{i}$ are related to the $\zeta$-polynomials $P_C(t) = \sum\limits _{i=0} ^{g + g^{\perp}} a_i t^{i}$, respectively, $P_{C^{\perp}} (t) = \sum\limits _{i=0} ^{ g + g^{\perp}} a_i ^{\perp} t^{i}$ by the equalities
(\ref{DuursmaByDzetaPolynomials}).

According to Proposition-Definition \ref{ExistsUniqueDzetaPolynomial}, Mac Williams  identities for $C$, $C^{\perp}$  are equivalent to the  functional equation (\ref{MacWilliamsDzetaPolynomial}) for $P_C(t)$, $P_{C^{\perp}} (t)$.
Making use of   (\ref{DuursmaByDzetaPolynomials}), one derives
\begin{align*}
D_C \left( \frac{1}{qt} \right) q^{g-1} t^{g + g^{\perp} -2} =
\left[ \frac{P_C \left( \frac{1}{qt} \right) - \frac{1}{q^g t^g}}{ \left(1 - \frac{1}{qt} \right) \left( 1 - \frac{1}{t} \right) } \right] q^{g-1}
t^{g + g^{\perp} -2} = \\
\left[ \frac{P_C \left( \frac{1}{qt} \right) - \frac{1}{q^g t^g} }{(qt-1)(t-1)} \right] q^g t^{g + g^{\perp}} =
\frac{P_C \left( \frac{1}{qt} \right) q^g t^{g + g^{\perp}} - t^{g^{\perp}}}{(1-t)(1-qt)} = \\
 \frac{P_{C^{\perp}}(t) - t^{g^{\perp}}}{(1-t)(1-qt)} = D_{C^{\perp}}(t),
\end{align*}
which reveals that (\ref{MacWilliamsDzetaPolynomial}) is equivalent to (\ref{MacWilliamsForDRP}).

Towards (\ref{DuursmasCoefficientsC}), let us denote $\rho _i := \binom{n}{d+i} (q-1) c_i$ for $0 \leq i \leq g + g^{\perp} -2$ and express (\ref{EnumeratorCByYZ}) in the form
\begin{align*}
\sum\limits _{i=0} ^{n - d - d^{\perp}} \rho _i y^{d+i} z^{n-d-i} =
\sum\limits _{s=0} ^n \left( \cW _C ^{(s)} - \cM _{n, n+1-k} ^{(s)} \right) (y+z) ^{n-s} y^s = \\
\sum\limits _{s=0} ^n \left( \cW _C ^{(s)} - \cM _{n, n+1-k} ^{(s)} \right) \left( \sum\limits _{i=0} ^{n-s} \binom{n-s}{i} y^{n-i} z^{i} \right).
\end{align*}
After changing the summation index of the left hand side to $j := n-d-i$ and exchanging the summation order of the right hand side, one obtains
\begin{align*}
\sum\limits _{j = d^{\perp}} ^{n-d} \rho _{n-d-j} y^{n-j} z^{j} =
\sum\limits _{i=0} ^n \left[ \sum\limits _{s=0} ^{n-i} \binom{n-s}{i} \left( \cW _C ^{(s)} - \cM _{n, n+1-k} ^{(s)} \right) \right] y^{n-i} z^{i}.
\end{align*}
The comparison of the coefficients of $y^{n-j} z^{j}$  in the above equality for $\forall d^{\perp} \leq j \leq n-d$ provides
$$
\rho _{n-d-j} = \sum\limits _{s=0} ^{n-j} \binom{n-s}{j} \left( \cW _C ^{(s)} - \cM _{n, n+1-k} ^{(s)} \right) \ \ \mbox{  for } \ \
\forall d^{\perp} \leq j \leq n-d.
$$
Changing the index to $i = n-d-j$, one derives (\ref{DuursmasCoefficientsC}).
Similar considerations on (\ref{EnumeratorCPerpByYZ}) yields (\ref{DuursmasCoefficientsCPerp}).

\end{proof}

In the case of $\FF_q$-linear codes, (\ref{DuursmasCoefficientsC}) and (\ref{DuursmasCoefficientsCPerp}) can be specified as follows:

\begin{corollary}   \label{DenominatorsCoeffDCLinear}
Let $C \subset \FF_q ^n$ be an $\FF_q$-linear code of minimum distance $d \geq 2$  with dual $C^{\perp} \subset \FF_q ^n$ of minimum distance $d^{\perp} \geq 2$ and
$$
D_C (t) = \sum\limits _{i=0} ^{g + g^{\perp}-2} c_i t^{i},  \ \ D_{C^{\perp}} (t) = \sum\limits _{i=0} ^{g + g^{\perp}-2} c_i ^{\perp} t^{i} \in \QQ[t]
$$
 be Duursma's reduced polynomials of $C$, $C^{\perp}$.
Then
\begin{equation}    \label{LinearDuursmasDenominators}
\binom{n}{d+i} c_i \in \ZZ \ \ \mbox{\rm  and } \ \ \binom{n}{d^{\perp}+i} c_i ^{\perp} \in \ZZ
\end{equation}
 are integers for all  $ 0 \leq i \leq g + g^{\perp} -2$.
\end{corollary}

\begin{proof}

Note that the $\FF_q$-linear structure of $C$, $C^{\perp}$ induces actions
$$
\FF_q ^* \times C \longrightarrow C, \ \ (\lambda, a) \mapsto (\lambda a_1, \ldots , \lambda a_n),
$$
respectively,
$$
\FF_q ^* \times C^{\perp} \longrightarrow C^{\perp}, \ \ (\lambda, b) \mapsto ( \lambda b_1, \ldots , \lambda b_n),
$$
preserving the Hamming weights.
The orbit spaces $\PP(C) :=  C \setminus \{ 0_{\FF_q} ^n \} / \FF_q ^*$, respectively, $\PP(C^{\perp}) := C^{\perp} \setminus \{ 0_{\FF_q} ^n \} / \FF_q ^*$ are the corresponding projectivizations, which inherit the Hamming weight
$$
{\rm wt} : \PP( \FF_q^n) = \PP ^{n-1}( \FF_q) \longrightarrow \{   1, \ldots , n \}, \ \
{\rm wt} ([a]) :=  \left| \{ 1 \leq i \leq n \, \vert \, a_i \neq 0 \} \right|
$$
  from the projectivization $\PP( \FF_q^n) = \PP ^{n-1}( \FF_q)$ of the ambient vector space $\FF_q ^n \supset C, C^{\perp}$.
The subsets $\PP(C)^{(s)} := \{ [a] \in \PP(C) \, \vert \, {\rm wt} ([a]) =s \}$,  $\PP(C^{\perp}) ^{(s)} := \{ [b] \in \PP(C^{\perp}) \, \vert \, {\rm wt} ([b]) =s \}$ are of cardinality
$$
\left| \PP(C) ^{(s)} \right| = \frac{\cW _{C} ^{(s)}}{q-1} \in \ZZ^{\geq 0}, \ \ \mbox{  respectively, }
 \left| \PP(C^{\perp}) ^{(s)} \right| = \frac{\cW _{C^{\perp}} ^{(s)}}{q-1} \in \ZZ^{\geq 0} \ \ \mbox{  for } \ \ \forall 1 \leq s \leq n.
$$
Baring in mind that the integers
$$
\cM _{n, n+1-k} ^{(s)} = \binom{n}{s} (q-1) \left[ \sum\limits _{i=0} ^{s-n-1+k} (-1) ^{i} \binom{s-1}{i} q^{s-d-i} \right],
$$
respectively,
$$
\cM _{n, k+1} ^{(s)} = \binom{n}{s} (q-1) \left[ \sum\limits _{i=0} ^{s-k-1} (-1) ^{i} \binom{s-1}{i} q^{s-d-i} \right]
$$
from (\ref{MDSWeightDistribution}) are divisible by $q-1$, one  makes use of  (\ref{DuursmasCoefficientsC}), (\ref{DuursmasCoefficientsCPerp})
in order to derive (\ref{LinearDuursmasDenominators}).

\end{proof}

%%\newpage
 %%%%%%%%%%%%%%%%%%%%%%%%%%%%%%%%%%%%%%%%%%%%%%%%%%%%%%%%%%%%%%%%%%%%%%%%%%%%%%%%%%%%%%%%%%%%%%%%%%%%%%%%%%%%%%%%%%%%%%%%%%%%%%%%%%%%%%%%%%%%%%%%%%%%%%%%%%%%%
\section{Mac Williams identities for  additive  codes are equivalent to Polarized Riemann-Roch Conditions}
 %%%%%%%%%%%%%%%%%%%%%%%%%%%%%%%%%%%%%%%%%%%%%%%%%%%%%%%%%%%%%%%%%%%%%%%%%%%%%%%%%%%%%%%%%%%%%%%%%%%%%%%%%%%%%%%%%%%%%%%%%%%%%%%%%%%%%%%%%%%%%%%%%%%%%%%%%%

 In \cite{KM2} we have shown that a formal power series $\zeta (t) = \sum\limits _{m=0} ^{\infty} \cA_m t^m \in \ZZ[[t]]$ is a quotient $\zeta (t) = \frac{P(t)}{(1-t)(1-qt)}$ of a polynomial $P(t) \in \ZZ[t]$ exactly when $\zeta (t)$ is subject  to the generic Riemann-Roch Conditions
 $$
 \cA_m = - q^{m+1} {\rm Res}_{\frac{1}{q}}  ( \zeta (t)) - {\rm Res} _1 ( \zeta (t)) \ \ \mbox{  for  } \ \ \forall m \geq \deg P(t) -1
 $$
 and the residuums ${\rm Res} _{\frac{1}{q}} ( \zeta (t))$, ${\rm Res} _1 ( \zeta (t))$ of $\zeta (t)$ at its simple poles $\frac{1}{q}$, $1$.
 By its very definition, the $\zeta$-function $\zeta _C(t) = \frac{P_C(t)}{(1-t)(1-qt)} = \sum\limits _{m=0} ^{\infty} \cA _m (C) t^m$ of an additive code
  $(C, +) < (G^n, +)$ is a quotient of a polynomial $P_C(t) \in \QQ[t]$, so that satisfies the Generic Riemann-Roch Conditions
  $$
  \cA_m (C) = - q^{m+1} {\rm ReS} _{\frac{1}{|G|}} ( \zeta _C(t)) - {\rm Res} _1 ( \zeta _C(t)) =
  \frac{q^{m+1} P_C \left( \frac{1}{|G|} \right) - P_C(1)}{|G|-1}
  $$
 for $\forall m \geq g + g^{\perp} -1,$  where $g$ is the genus of $C$ and $g^{\perp}$ is the genus of $C^{\perp}$.
 The present section defines a more refined, polarized form of Riemann-Roch conditions and establishes the equivalence of the Mac Williams identities for $C$, $C^{\perp}$ to the polarized Riemann-Roch conditions on their $\zeta$-functions.

The following lemma   motivates the notion of Riemann-Roch Conditions for a formal power series of one variable.

\begin{lemma}   \label{RRCForCurves}
Let $X / \FF_q \subset \PP^N ( \overline{\FF_q})$ be a smooth irreducible   curve of genus $g$, defined over a finite field
 $\FF_q$ and $\zeta _X (t) = \sum\limits _{m=0} ^{\infty} \cA _m (X) t^m$ be the local  Weil    $\zeta$-function of $X$.
Then the Riemann-Roch Theorem  on $X$ implies the Riemann-Roch Conditions
$$
\cA_m (X) = q^{m-g+1} \cA _{2g-2-m} (X) +    (q^{m-g+1}-1)  {\rm Res} _1 ( \zeta _X(t)) \ \ \mbox{  for  } \ \ \forall m \geq g,
$$
 where $\cA_m (X)$ is   the number  of the effective divisors  of degree $m$ of the function field  $\FF_q (X)$  of $X$ over $\FF_q$   and  ${\rm Res} _1 ( \zeta _X(t))$ is the residuum
  of $\zeta _X(t)$ at $t=1$.
\end{lemma}

\begin{proof}

For an arbitrary divisor $G$ of the function field $\FF_q (X)$, let $H^0 (X, \cO _X ( [G]))$ be the space of the global sections of the line bundle, associated with $G$ and $l(G) := \dim _{\FF_q} H^0 (X, \cO _X ([G]))$.
Riemann-Roch Theorem  asserts  the existence of a  canonical divisor $K_X$ of degree $\deg K_X = 2g-2$ with
\begin{equation}     \label{RRTh}
l(G) = l(K_X -G) + \deg G -g+1
\end{equation}
for all divisors $G$ of $\FF_q(X)$.
In particular, if $\deg G > 2g-2$ then
$$
l(G) = \deg G -g +1.
$$

For any $k \in \NN$ let $X( \FF_{q^k}) := X \cap \PP^N ( \FF_{q^k})$ be the set of the $\FF_{q^k}$-rational points on $X$.
The formal power series
$$
\zeta _X(t) := \exp \left( \sum\limits _{k=1} ^{\infty} \left| X( \FF_{q^k} ) \right| \frac{t^k}{k} \right) \in \CC [[t]]
$$
is called the local Weil $\zeta$-function of $X$.
It is well known (cf.Theorem 4.1.11 from \cite{NX}) that there is a $\zeta$-polynomial $P_X(t) \in \ZZ[t]$ of $\deg P_X(t) = 2g$, such that
$$
\zeta _X(t) = \frac{P_X(t)}{(1-t)(1-qt)}
$$
  and the residuum
$$
{\rm Res} _1 ( \zeta _X(t)) = \frac{P_X(1)}{q-1} = \frac{h}{q-1}
$$
 for the class number  $h$  of $\FF_q (X)$.
If $G_1, \ldots , G_h$ is a complete set of representatives of the linear equivalence classes of the divisors of $\FF_q (X)$ of degree $m \in \ZZ^{\geq 0}$ then
the divisors  $K_X - G_1, \ldots , K_X - G_h$ form  a complete set of representatives of the linear equivalence classes   of $\FF_q (X)$ of degree
 $\deg (K_X - G_i) = 2g-2-m$.
The effective divisors of $\FF_q (X)$, which are linearly equivalent to $G_i$ constitute  the projective space
 $\PP ( H^0 (X, \cO _X ( [G_i]))) = \PP ^{l(G_i) -1} ( \FF_q)$.
Thus, the number of the effective divisors of $\FF_q (X)$ of degree $m$ is
\begin{equation}   \label{AX}
\cA _m (X) = \sum\limits _{i=1} ^h \left| \PP^{l (G_i) -1} ( \FF_q) \right| = \sum\limits _{i=1} ^h \frac{q^{l(G_i)} -1}{q-1}.
\end{equation}
Substituting (\ref{RRTh}) in (\ref{AX}), one obtains
$$
\cA _m (X) = q^{m-g+1} \sum\limits _{i=1} ^h \left( \frac{q^{l (K_X - G_i)} -1}{q-1} \right) + h \left( \frac{q^{m-g+1} -1}{q-1} \right).
$$
Bearing in mind that
$$
\sum\limits _{i=1} ^h \left( \frac{q^{l(K_X - G_i)} -1}{q-1} \right) = \cA _{2g-2-m} (X),
$$
one concludes that
\begin{equation}   \label{RRC_X}
\cA _m (X) = q^{m-g+1} \cA _{2g-2-m} (X) + h \left( \frac{ q^{m-g+1}-1}{q-1} \right) \ \ \mbox{  for } \ \ \forall m \geq 0.
\end{equation}
Note that in the case of $g \geq 2$ the relations (\ref{RRC_X}) with  $0 \leq m \leq g-2$   are equivalent to
 are equivalent to the ones with index $g \leq 2g-2-m \leq 2g-2$   and (\ref{RRC_X}) is trivial for $m = g-1$.
If $g=0$ then $X = \PP^1 ( \overline{\FF_q})$ is the projective line and the equalities  (\ref{RRC_X}) reduce to
$$
\cA_m (X) = \frac{q^{m+1}-1}{q-1} \ \ \mbox{ for } \ \ \forall m \geq 0.
$$
When  $g=1$, the curve $X$ is elliptic and
$$
\cA_m (X) = h \left( \frac{q^m-1}{q-1} \right)
$$
 for the class number $h$ and all $m \geq 1$.

\end{proof}

\begin{definition}    \label{RRC}
A formal power series $\zeta (t) = \sum\limits _{m=0} ^{\infty} \cA _m t^{m} \in {\mathbb C}[[t]]$ satisfies the Riemann-Roch Conditions ${\rm RRC}_q (g)$ with base $q \in \NN$ of genus $g \in {\mathbb  Z}^{\geq 0}$ if
$$
\cA _m = q^{m-g+1} \cA_{2g-2-m} +     (q^{m-g+1}-1)  {\rm Res} _1 ( \zeta (t))   \ \ \mbox{\rm  for  } \ \ \forall m \geq g
$$
and the residuum ${\rm Res} _1 ( \zeta (t))$ of $\zeta (t)$ at $t=1$.
\end{definition}

Here is a polarized version of the Riemann-Roch Conditions.

\begin{definition}   \label{PRRC}
Formal power series $\zeta (t) = \sum\limits _{m=0} ^{\infty} \cA _m t^m \in \CC[[t]]  $ and
$\zeta ^{\perp} (t) = \sum\limits _{m=0} ^{\infty}\cA ^{\perp} _m t^m  \in \CC[[t]] $ satisfy the Polarized Riemann-Roch Conditions ${\rm PRRC}_q (g, g^{\perp})$   of genera $g, g^{\perp} \in {\mathbb Z} ^{\geq 0}$ with base $q \in \NN$ if
$$
\cA_m = q^{m-g+1} \cA ^{\perp} _{g + g^{\perp} -2-m} +    (q^{m-g+1}-1) {\rm Res} _1 ( \zeta (t)) \ \ \mbox{\rm  for  }  \ \ \forall m \geq g,
$$
$$
\cA _{g-1} = \cA ^{\perp} _{g^{\perp} -1} \ \ \mbox{\rm  and  }
$$
$$
\cA ^{\perp} _m = q^{m-g^{\perp}+1}  \cA _{g + g^{\perp} -2-m} +   (q^{m - g^{\perp} +1}-1)  {\rm Res} _1 ( \zeta ^{\perp} (t)) \ \
\mbox{\rm  for } \ \ \forall m \geq g^{\perp},
$$
where ${\rm Res} _1 ( \zeta (t))$, $ {\rm Res} _1 ( \zeta ^{\perp} (t))$ stand for the corresponding  residuums at $t=1$.
\end{definition}

One can view the  Riemann-Roch Theorem  on a smooth irreducible projective curve  $X / \FF_q \subset \PP^N ( \overline{\FF_q})$ as a quantitative expression of the Serre duality on $X$.
Thus,  ${\rm RRC}_q (g)$ and, therefore, ${\rm PRRC}_q (g, g^{\perp})$   may be interpreted as a Serre duality between the formal power series
 $\zeta(t)$, $\zeta ^{\perp}(t)$.

 Observe also  that ${\rm PRRC}_q (g, g^{\perp})$ implies
$$
\cA _m = \kappa_1 q^m + \kappa_2, \ \ \cA _m ^{\perp} = \kappa_1 ^{\perp} q^m + \kappa_2 ^{\perp} \ \ \mbox{ for } \ \ \forall m \geq g + g^{\perp} -1
$$
 and some constants $\kappa_j, \kappa_j ^{\perp} \in \CC$.
These are  equivalent to the recurrence relations
$$
\cA _{m+2}  - (q+1) \cA _{m+1} + q \cA_m = \cA _{m+2} ^{\perp}  - (q+1) \cA _{m+1} ^{\perp} + q \cA_m ^{\perp} =  0 \ \ \mbox{  for } \ \
\forall m \geq g + g^{\perp} -1
$$
and hold exactly when
$$
\zeta (t) = \frac{P(t)}{(1-t)(1-qt)} \ \ \mbox{  and  } \ \ \zeta ^{\perp} (t) = \frac{P^{\perp}(t)}{(1-t)(1-qt)}
$$
 for polynomials $P(t)$, $P^{\perp}(t)$.
 Thus, the Polarized Riemann-Roch Conditions with base $q$    require $\zeta (t)$, $\zeta ^{\perp} (t)$ to be rational functions of $t$ with denominators $(1-t)(1-qt)$ and imply the Generic Riemann-Roch Conditions on $\zeta (t)$ and on $\zeta ^{\perp}(t)$.
 Note that the Generic Riemann-Roch Conditions for $\zeta (t) = \frac{P(t)}{(1-t)(1-qt)}$ coincide with the Riemann-Roch Conditions ${\rm RRC} _q (g)$ with $m \geq 2g-1$ if and only if 
 $$
 P \left( \frac{1}{q} \right) = \left( \frac{1}{q} \right) ^g P(1).
 $$

%%%% {\sc START SEPT 16, 2016}

  %%%%%%%%%%%%%%%%%%%%%%%%%%%%%%%%%%%%%%%%%%%%%%%%%%%%%%%%%%%%%555555

While preparing the present article, we came up with Randriambololona's article \cite{Hugues} on Harder-Narasimhan theory, Serre duality and Riemann-Roch Theorem for linear codes.
Our main Theorem \ref{MWIffPRRC} reveals that the Riemann-Roch Theorem 44 from \cite{Hugues} is stronger than our Polarized Riemann-Roch Conditions ${\rm PRRC} _q (g, g^{\perp})$ from Definition \ref{PRRC}.
More precisely, for an arbitrary $\FF_q$-linear $[n,k,d]$-code $C \subset \FF_q ^n$ of genus $g = n+1 -d-k$ and an arbitrary subset
$J \subseteq [n] := \{ 1, \ldots , n \}$,  Randriambololona defines the cohomology group $H^0 (C, J) := C \cap \FF_q ^J$ as the  largest linear subspace of $C$ with support $J$.
If $C^{\perp} \subset \FF_q ^n$ is the dual code of $C$ and $[n] \setminus J$ is the complement of $J$, then the Riemann-Roch Theorem 44 from \cite{Hugues} asserts that
\begin{equation}     \label{HuguesRRT}
\dim _{\FF_q} H^0 (C,J) - \dim _{\FF_q} H^0 (C^{\perp}, [n] \setminus J)= (|J| -d) -g +1.
\end{equation}
The author mentions that (\ref{HuguesRRT}) implies the functional equation
\begin{equation}   \label{EEDzetaFunctions}
\zeta _{C^{\perp}} (t) = \zeta _C \left( \frac{1}{qt} \right) q^{g-1} t^{g + g^{\perp} -2}
\end{equation}
of Duursma's $\zeta$-functions $\zeta _C(t)$, $\zeta _{C^{\perp}} (t)$ of $C, C^{\perp}$.
As far as (\ref{EEDzetaFunctions}) is equivalent to Mac Williams identities for the weight distribution of $C, C^{\perp}$, our Theorem \ref{MWIffPRRC} reveals that  Randriambololona's Riemann-Roch Theorem 44 from \cite{Hugues} implies our Polarized Riemann-Roch Conditions ${\rm PRRC} _{q} (g, g^{\perp})$
 on $\zeta _C(t)$, $\zeta _{C^{\perp}} (t)$.

%%%%%%%%%%%%%%%%%%%%%%%%%%%%%%%%%%%%%%%%%%%%%%%%%%%%%%%%%%%%%%%%%%%%%%

Here is the main result of the present article, which    interprets  Mac Williams duality on additive codes as Polarized Riemann-Roch Conditions or as a polarized form of the Serre duality on a smooth irreducible projective curve, defined over a finite field.

\begin{theorem}   \label{MWIffPRRC}
Mac Williams identities for the weight distribution of an additive code $(C, +) < (G^n, +)$ of minimum distance $d \geq 2$ and genus $g$ with dual
$(C^{\perp}, .) < (\widehat{G}^n, .)$ of minimum distance $d^{\perp} \geq 2$ and genus $g^{\perp}$ are equivalent to  the Polarized Riemann-Roch Conditions
 ${\rm PRRC} _{|G|} (g, g^{\perp})$ on their  $\zeta$-functions
  $\zeta _C(t)$, $\zeta _{C^{\perp}}(t)$.
\end{theorem}

%%%%%%%%%%%%%%%%%%%%%%%%%%%%%%%%%%%%%%%%%%%%%%%%%%%%%%%%%%%%%%%%%%%%%%%%%%%%%%%%%%%%%%%%%%%%%%%%%%%%%%%%%%%%%%%%%%%%

\begin{proof}

Let us denote by $q := |G|$ the order of $G$.
First, we prove the theorem for  $g, g^{\perp} \in \NN$.
If $\zeta _C(t) := \frac{P_C(t)}{(1-t)(1-qt)} = \sum\limits _{m=0} ^{\infty} \cA_m (C) t^m$ for some $\cA_m (C) \in \QQ$ then
\begin{equation}     \label{D_And_A}
D_C(t) = \zeta _C(t) - \frac{t^g}{(1-t)(1-qt)} =
 \sum\limits _{m=0} ^{\infty} \cA _m (C) t^m - \sum\limits _{m=g} ^{\infty} \left( \frac{q^{m-g+1}-1}{q-1} \right) t^m
\end{equation}
by
\begin{align*}
\frac{t^g}{(1-t)(1-qt)} =  t^g \left( \sum\limits _{i=0} ^{\infty} t^{i} \right) \left( \sum\limits _{j=0} ^{\infty} q^{j} t^{j} \right) =  \\
\sum\limits _{m=g} ^{\infty} ( q^{m-g} + q^{m-g-1} + \ldots + q +1) t^m =
\sum\limits _{m=g} ^{\infty} \left( \frac{q^{m-g+1} -1}{q-1} \right) t^m \in \ZZ [[t]].
\end{align*}
Thus, $D_C(t) - \zeta _C(t)$ is a sum of terms of degree $\geq g$,
\begin{equation}   \label{FirstFormulaDC}
\begin{split}
D_C (t) = \sum\limits _{m=0} ^{g-2} \cA _m (C) t^m + \cA _{g-1} (C) t^{g-1} + \sum\limits _{i=g} ^{g + g^{\perp} -2} c_i t^{i} =  \\
= \sum\limits _{m=0} ^{g-2} \cA _m (C) t^m + \cA _{g-1} (C) t^{g-1} +
\left( \sum\limits _{m=0} ^{g^{\perp}-2} c_{g + g^{\perp} -2-m} t^{-m} \right) t^{g + g^{\perp} -2}
\end{split}
\end{equation}
and, respectively,
$$
D_{C^{\perp}} (t) = \sum\limits _{m=0} ^{g^{\perp} -2} \cA _m ( C^{\perp} ) t^m + \cA _{g^{\perp} -1} (C^{\perp}) t^{g^{\perp} -1} +
\left( \sum\limits _{m=0} ^{g-2} c^{\perp} _{g + g^{\perp} -2-m} t^{-m} \right) t^{g + g^{\perp} -2}
$$
for $\zeta _{C^{\perp}} (t) := \frac{P_{C^{\perp}} (t)}{(1-t)(1-qt)} = \sum\limits _{m=0} ^{\infty} \cA _m ( C^{\perp}) t^m$.
According to
\begin{align*}
D_C \left( \frac{1}{qt} \right) q^{g-1} t^{g + g^{\perp} -2} =  \\
= \sum\limits _{m=0} ^{g-2} \cA _m (C) q^{ g-1-m} t^{g + g^{\perp} -2-m} + \cA _{g-1} (C) t^{g^{\perp}-1} +
\sum\limits _{m=0} ^{ g^{\perp} -2} c_{g + g^{\perp} -2-m} q^{m - g^{\perp} +1} t^m = \\
= \sum\limits _{m=0} ^{ g^{\perp} -2} c_{g + g^{\perp} -2-m} q^{ m - g^{\perp} +1} t^m + \cA _{g-1} (C) t^{g^{\perp} -1} +
\left( \sum\limits _{m=0} ^{ g-2} q^{ -m+g-1} \cA _m (C) t^{-m} \right) t^{g + g^{\perp} -2},
\end{align*}
Mac Williams identities (\ref{MacWilliamsForDRP}) for Duursma's reduced polynomials of a pair $C, C^{\perp} \subset \FF_q ^n$ of mutually dual linear codes of genus $g \geq 1$, respectively,  $g^{\perp} \geq 1$ amount to
\begin{equation}    \label{II_Part_D}
c_{g + g^{\perp} -2-m} = q^{ -m + g^{\perp} -1} \cA _m (C^{\perp}) \ \ \mbox{  for  } \ \ \forall 0 \leq m \leq g^{\perp} -2,
\end{equation}
\begin{equation}   \label{MW_Middle}
\cA _{g^{\perp} -1} (C^{\perp}) = \cA _{g-1} (C) \ \ \mbox{  and }
\end{equation}
\begin{equation}   \label{II_Part_DBar}
c^{\perp} _{g + g^{\perp} -2-m} = q^{ -m + g -1} \cA _m (C) \ \ \mbox{  for } \ \ \forall 0 \leq m \leq g-2.
\end{equation}
Substituting $m = g + g^{\perp} -2-i$ and making use of (\ref{D_And_A}), one observes that (\ref{II_Part_D}) is equivalent to
$$
\cA _i (C) = q^{ i -g+1} \cA _{g + g^{\perp} -2-i} ( C^{\perp}) + \left( \frac{q^{i-g+1} -1}{q-1} \right) \ \ \mbox{  for  } \ \
\forall g \leq i \leq g + g^{\perp} -2.
$$
Exchanging $C$ with $C^{\perp}$, one expresses (\ref{II_Part_DBar}) in the form
$$
\cA _i (C^{\perp}) = q^{i - g^{\perp} +1} \cA _{g + g^{\perp} -2-i} (C) + \left( \frac{q^{i - g^{\perp} +1}-1}{q-1} \right) \ \ \mbox{  for } \ \
\forall g^{\perp} \leq i \leq g + g^{\perp} -2.
$$
According to (\ref{D_And_A}),
$$
\cA_i (C) = \frac{q^{i-g+1}-1}{q-1} \ \ \mbox{ for } \ \ \forall i \geq g + g^{\perp} -1.
$$
Similarly,
$$
\cA _i (C^{\perp}) = \frac{q^{i-g^{\perp}+1}-1}{q-1} \ \ \mbox{  for } \ \ \forall i \geq g + g^{\perp} -1.
$$
Bearing in mind that $\zeta _C(t)$ and $\zeta _{C^{\perp}}(t)$ have no pole at $t=0$, one introduces $\cA _{-j} (C) = \cA _{-j} (C^{\perp}) =0$
 for $\forall j \in \NN$ and  expresses Mac Williams identities in the form
\begin{equation}   \label{MWA}
\cA _i (C) = q^{i -g+1} \cA _{g + g^{\perp} -2-i} (C^{\perp}) + \left( \frac{q^{i-g+1}-1}{q-1} \right) \ \ \mbox{ for  } \ \
\forall i \geq g,
\end{equation}
\begin{equation}    \label{MWB}
\cA _{g^{\perp} -1} (C^{\perp}) = \cA _{g-1} (C) \ \ \mbox{  and  }
\end{equation}
\begin{equation}   \label{MWC}
\cA _i (C^{\perp}) = q^{i - g^{\perp} +1} \cA _{g + g^{\perp} -2-i} (C) + \left( \frac{q^{i - g^{\perp} +1}-1}{q-1} \right) \ \ \mbox{  for } \ \
\forall i \geq g^{\perp}.
\end{equation}
Note also that the rational function
$$
\zeta _C (t) = \frac{P_C(t)}{(1-t)(1-qt)}
$$
 has residuum
 $$
 {\rm Res} _1 ( \zeta _C(t)) = \frac{P_C(1)}{q-1} = \frac{1}{q-1}
 $$
  at $1$.
   Thus, for $g \geq 1$ and $g^{\perp} \geq 1$,  Mac Williams identities (\ref{MWA}), (\ref{MWB}), (\ref{MWC}) for   $C, C^{\perp} \subset \FF_q ^n$ are equivalent to the polarized Riemann-Roch conditions ${\rm PRRC} (g, g^{\perp})$.

   In the case of $g=0$, one has $|C| = |G| ^{n+1-d} < |G|^n$ by $d \geq 2$.
  Thus, Lemma \ref{MDSDualIdMDS} applies to provide $g^{\perp} =0$.
  The $\zeta$-functions
  $$
  \zeta _C(t) = \zeta _{C^{\perp}} (t) = \frac{1}{(1-t)(1-qt)} = \zeta _{\PP^1 ( \overline{\FF_q})} (t)
  $$
  coincide with the $\zeta$-function of the projective line $\PP^1 ( \overline{\FF_q})$ and satisfy the Riemann-Roch Conditions ${\rm RRC} _{|G|} (0)$ of genus $0$, which are equivalent to the Polarized Riemann-Roch Conditions ${\rm PRRC} _{|G|} (0,0)$.

   The $\zeta$-functions
 $$
   \zeta _C(t) = \zeta _{C^{\perp}} (t) = \frac{1}{(1-t)(1-qt)} = \zeta _{\PP^1 ( \overline{\FF_q})} (t)
   $$
   coincide with the $\zeta$-function of the projective line $\PP^1 ( \overline{\FF_q})$ and satisfy the Riemann-Roch Conditions ${\rm RRC} (0)$ of genus $g=0$,
   which are equivalent to the Polarized Riemann-Roch Conditions ${\rm PRRC} (0, 0)$.

\end{proof}

%%%%%%%%%%%%%%%%%%%%%%%%%%%%%%%%%%%%%%%%%%%%%%%%%%%%%%%%%%%%%%%%%%%%%%%%%%%%%%%%%%%%%%%%%%%%%%%%%%%%%%%%%%%%%%%%%%%%%%%%%%%%%%%%%%%%%%%%%%
As a byproduct of the proof of Theorem \ref{MWIffPRRC} we obtain the following

\begin{corollary}    \label{LowerPartD}
The lower parts $\varphi _C (t) = \sum\limits _{i=0} ^{g-2} c_i t^{i} $,
$\varphi _{C^{\perp}} (t) = \sum\limits _{i=0} ^{g^{\perp}-2} c_i ^{\perp} t^{i} $
of Duursma's reduced polynomials $D_C(t) = \sum\limits _{i=0} ^{g + g^{\perp} -2} c_i t^{i}$,
$D_{C^{\perp}} (t) = \sum\limits _{i=0} ^{g + g^{\perp} -2} c_i ^{\perp} t^{i}$ of an  additive code $(C,+) < (G^n, +)$  of minimum distance $d \geq 2$ and genus $g \geq 1$ and its dual $(C^{\perp}, .) < (\widehat{G}^n, .)$ of minimum distance $d^{\perp} \geq 2$ and genus $g^{\perp} \geq 1$, together with the number
 $c_{g-1} = c^{\perp} _{g^{\perp}-1} \in \QQ$ determine uniquely
\begin{equation}   \label{DCFormula}
D_C(t) = \varphi _C(t) + c_{g-1} t^{g-1} + \varphi _{C^{\perp}} \left( \frac{1}{qt} \right) q^{g^{\perp}-1} t^{g + g^{\perp} -2},
\end{equation}
\begin{equation}    \label{DCBarFormula}
D_{C^{\perp}} (t) = \varphi _{C^{\perp}} (t) + c_{g-1} t^{g^{\perp}-1} + \varphi _C \left( \frac{1}{qt} \right) q^{g-1} t^{g + g^{\perp}-2}.
\end{equation}
\end{corollary}

\begin{proof}

The substitution of (\ref{II_Part_D}) in (\ref{FirstFormulaDC}) yields
$$
D_C(t) = \varphi _C(t) + c_{g-1} t^{g-1} + \left( \sum\limits _{m=0} ^{g^{\perp}-2} c^{\perp} _m q^{-m} t^{-m} \right) q^{g^{\perp}-1} t^{ g + g^{\perp}-2},
$$
whereas (\ref{DCFormula}).
Replacing $C$ by $C^{\perp}$, $C^{\perp}$ by $C$ and $c^{\perp} _{g^{\perp}-1}$ by $c_{g-1}$, one obtains (\ref{DCBarFormula}).

\end{proof}

%%\newpage
%%%%%%%%%%%%%%%%%%%%%%%%%%%%%%%%%%%%%%%%%%%%%%%%%%%%%%%%%%%%%%%%%%%%%%%%%
\section{ Averaging, algebraic-geometric and probabilistic  \\ interpretations of the coefficients of Duursma's  \\ reduced polynomial }

Let $G$ be a finite abelian group, $(C, +) \leq (G^n, +)$ be an additive code.
Abbreviate $[n] := \{ 1, \ldots , n \}$ and denote by $\binom{[n]}{s}$ the collection of the subsets $\alpha = \{ \alpha _1, \ldots , \alpha _s \} \subseteq [n]$ of cardinality $1 \leq s \leq n$.
 We proceed  by an averaging interpretation  of the lower parts of Duursma's reduced polynomials of $C$ and $C^{\perp}$.

\begin{proposition}    \label{CoeffDCAreAverages}
Let $(C, +) < (G^n, +)$ be an additive code of minimum distance $d \geq 2$ and genus $g \geq 1$  with dual $(C^{\perp}, .) < ( \widehat{G}^n, .)$ of minimum distance $d^{\perp} \geq 2$ and genus $g^{\perp} \geq 1$.
Denote by $D_C (t) = \sum\limits _{i=0} ^{g + g^{\perp}-2} c_i t^{i}$, $D_{C^{\perp}} (t) = \sum\limits _{i=0} ^{g + g^{\perp}-2} c_i ^{\perp} t^{i} \in \QQ[t]$
Duursma's reduced polynomials of these codes and put
$$
(C \setminus \{ 0_G ^n \} ) ^{( \subseteq \gamma)} := \{ a \in C \setminus \{ 0_G ^n \} \, \vert \, {\rm Supp} (a) \subseteq \gamma \},
$$
respectively,
$$
(C^{\perp} \setminus \{ \varepsilon ^n \} ) ^{( \subseteq \gamma)} := \{ \pi \in C^{\perp} \setminus \{ \varepsilon ^n \} \, \vert \, {\rm Supp} ( \pi) \subseteq \gamma \}
$$
for $\gamma \in \binom{[n]}{s}$.
Then
\begin{equation}    \label{AverageC}
(|G|-1) c_i = \binom{n}{d+i} ^{-1} \left( \sum\limits _{\gamma \in \binom{[n]}{d+i}} \left| \left( C \setminus \{ 0_G ^n \} \right) ^{(\subseteq \gamma)} \right| \right) \ \ \mbox{  for } \ \ \forall 0 \leq i \leq g-1
\end{equation}
is the average cardinality of an intersection of $C \setminus \{ 0_G ^n \}$ with $n-d-i$ coordinate hyperplanes in $(G^n, +)$ and
\begin{equation}   \label{AverageCPerp}
(|G|-1) c_i ^{\perp} = \binom{n}{d^{\perp}+i} ^{-1} \left( \sum\limits _{\gamma \in \binom{[n]}{d^{\perp}+i}}  \left| \left( C^{\perp} \setminus \{ \varepsilon ^n \} \right) ^{( \subseteq \gamma)} \right| \right) \ \ \mbox{\rm  for } \ \ \forall 0 \leq i \leq g^{\perp} -1
\end{equation}
is the average cardinality of an intersection of $C^{\perp}$ with $n-d^{\perp}-1$ coordinate hyperplanes  in $(\widehat{G}^n, .)$.
\end{proposition}

\begin{proof}

The equality (\ref{AverageC}) will be derived by counting the disjoint union
$$
U^{(d+i)} := \coprod\limits _{\gamma \in \binom{[n]}{d+i}} \left( C \setminus \{ 0_G ^n \} \right) ^{( \subseteq \gamma)}
$$
in two different ways.
Namely, a word $a \in C \setminus \{ 0_G ^n \}$ of weight ${\rm wt} (a) = s \in \NN$ has support $\sigma := {\rm Supp} (a) \subseteq \gamma$ for some
 $\gamma \in \binom{[n]}{d+i}$ if and only if the complements $\neg \gamma := [n] \setminus \gamma \subseteq [n] \setminus \sigma =: \neg \sigma \in \binom{[n]}{n-s}$ are subject to the opposite inclusion.
There are exactly $\binom{n-s}{n-d-i}$ such $\neg \gamma$ for $\left| \neg \gamma \right| = n-d-i \leq n-s \left| \neg \sigma \right|$ and none of such
 $\neg \gamma$ for $n-d-i > n-s$.
Thus, any $a \in C \setminus \{ 0_G ^n \}$ of ${\rm wt} (a) = s \leq d+i$ is counted $\binom{n-s}{n-d-i}$ times in $U^{(d+i)}$ and noone
 $a \in C \setminus \{ 0_G^n \}$ with ${\rm wt} (a) = s > d+i$ is counted in $U^{( d+i)}$.
That justifies the equality
$$
\left| U^{(d+i)} \right| = \sum\limits _{s=1} ^{d+i} \binom{n-s}{n-d-i} \cW _C ^{(s)} \ \ \mbox{  for } \ \ \forall 0 \leq i \leq n-d.
$$
According to (\ref{DuursmasCoefficientsC}) from Proposition \ref{WCByDCAnd CoeffDenominators},
$$
(|G|-1) \binom{n}{d+i} c_i = \sum\limits _{s=1} ^{d+i} \cW _C ^{(s)} \binom{n-s}{n-d-i} \ \ \mbox{  for } \ \ \forall d \leq d + i \leq n-k,
$$
due to the vanishing of $\cM _{n,n+1-k} ^{(s)} =0$ for $\forall 1 \leq s \leq d+i \leq n-k$.
As a result,
$$
\left| U^{(d+i)} \right| = (|G|-1) \binom{n}{d+i} c_i \ \ \mbox{ for } \ \ 0 \leq i \leq n-k-d = g-1.
$$
Combining with
$$
\left| U^{(d+i)} \right| = \sum\limits _{\gamma \in \binom{[n]}{d+i}} \left| \left( C \setminus \{ 0_G^n \} \right) ^{( \subseteq \gamma)} \right| \ \
\mbox{  for } \ \  0 \leq i \leq n-d,
$$
one justifies (\ref{AverageC}) for $\forall 0 \leq i \leq g-1$.
Similar considerations on $\coprod\limits _{\gamma \in \binom{[n]}{d+i}} \left( C^{\perp} \setminus \{ 0_G ^n \} \right) ^{( \subseteq \gamma)}$ provide (\ref{AverageCPerp}) for $\forall 0 \leq i \leq g^{\perp}-1$.

\end{proof}

In the case of $\FF_q$-linear codes $C$, $C^{\perp} \subset \FF_q ^n$, when the finite  sets $(C \setminus \{ 0 _{\FF_q} ^n \} ) ^{( \subseteq \gamma)}$,
$(C^{\perp} \setminus \{ 0_{\FF_q} ^n \} ) ^{( \subseteq \gamma)}$ are invariant under componentwise multiplications by $\lambda \in \FF_q ^*$,
the following corollary formulates the result in terms of the cardinalities of the subsets
$$
\PP (C)^{( \subseteq \gamma)} := \{ [a] \in \PP(C) \, \vert \, {\rm Supp} ([a]) \subseteq \gamma \},
$$
respectively,
$$
\PP(C^{\perp}) ^{( \subseteq \gamma)} := \{ [a] \in \PP (C^{\perp}) \, \vert \, {\rm Supp} ([a]) \subseteq \gamma \}
$$
of the projectivizations $\PP(C) := C \setminus \{ 0_{\FF_q} ^n \} / \FF_q ^* \subset \PP ( \FF_q^n) := \FF_q ^n \setminus \{ 0_{\FF_q} ^n \} / \FF_q ^*$,
respectively, $\PP (C^{\perp}) := C^{\perp} \setminus \{ 0_{\FF_q} ^n \} / \FF_q ^* \subset \PP ( \FF_q ^n )$, viewed as projective subspaces of the projectivization $\PP( \FF_q ^n)$ of $\FF_q^n$.

Pellikaan, Shen and van Wee have shown in \cite{PShW} that an arbitrary $\FF_q$-linear code $C \subset \FF_q ^n$ has an algebraic-geometric representation.
It means an existence of a smooth irreducible projective curve $X / \FF_q \subset \PP^N ( \overline{\FF_q})$, defined over $\FF_q$, distinct $\FF_q$-rational points $P_1, \ldots , P_n \in X( \FF_q) := X \cap \PP ^N ( \FF_q)$ and a divisor $G$ of the function field $\FF_q(X)$, such that ${\rm Supp} (G) \cap {\rm Supp} (D) = \emptyset$ and $C = \cE _D \cL _X (G)$ is the image of the evaluation map 
$$
\cE_D : \cL _X(g) = H^0 (X, \cO _X ([G])) \longrightarrow \FF_q ^n, \ \ \cE_D (f) = (f(P_1), \ldots , f(P_n)) 
$$
at $D := P_1 + \ldots + P_n$.
Any algebraic-geometric realization $C = \cE _D \cL _X (G)$ of $C$ is associated with a algebraic-geometric realization $C^{\perp} = \cE _D \cL _X (K_X - G +d)$ of the dual codes $C^{\perp} \subset \FF_q ^n$, where $K_X$ stands for a canonical divisor of $\FF_q(X)$.
Form now on, we denote by $l(E) := \dim _{\FF_q} \cL _X(E)$ the dimension of $\cL_X(E)$.
The next proposition interprets the elements of the projectivizations 
$$
\PP(C) := C \setminus \{ 0_{\FF_q} ^n \} / \FF_q ^* \ \  \mbox{ rm and } \ \ 
\PP(C) ^{( \subseteq \gamma)} := \PP \left( \left( C \setminus \{ 0_{\FF_q} ^n \} \right) \right)
$$
 with $\gamma \in \binom{[n]}{s}$, $d \leq s \leq n-k$ as orbits of effective divisors of $\FF_q (X)$.

\begin{corollary}    \label{ProjInterpretationOrbitsOfEffectiveDivisors}
 Let $C = \cE_D \cL _X (G)$ be an algebraic-geometric representation of an $\FF_q$-linear code $C \subset \FF_q ^n$ of minimum distance $d \geq 2$ and genus $g 
 \geq 1$ with dual $C^{\perp} = \cE_D \cL _X (K_X - G +d)$ of minimum distance $d^{\perp} \geq 2$ and genus $g^{\perp} \geq 1$.
 Denote by
 $$
 {\rm Div} _{\geq 0} ( \sim G, D) := \{ E = G + {\rm div} (f) \geq 0 \, \vert \, f \in \FF_q (X), \, {\rm Supp} (D) \nsubseteq {\rm Supp} (E) \}
 $$
 the set of the effective divisors of $\FF_q(X)$, which are linearly equivalent  to $G$ and do not contain $\{ P_1, \ldots , P_n \}$ in its support and put 
 $D_{\delta} := \sum\limits _{i \in \delta} P_i$ for $\forall  \delta \in \binom{[n]}{s}$, $1 \leq s \leq n$.

 (i) Then the kernel $\ker \cE_D = \cL _X (G - D)$ of the surjective $\FF_q$-linear evaluation map $\cE_D : \cL_X (G) \rightarrow C$ acts on ${\rm Div} _{\geq 0} ( \sim G, D)$ by the rule
 \begin{equation}   \label{KernelAction}
 \cL _X(G-D) \times {\rm Div} _{\geq 0} ( \sim G, D) \longrightarrow {\rm Div} _{\geq 0} ( \sim G, D), \ \ 
 (g, G = {\rm div} (f)) \mapsto G +  {\rm div} (f+g)
 \end{equation}
 and the projectivization
 $$
 \PP(C) = {\rm Div} _{\geq 0} ( \sim G, D) / \cL _X (G-D)
 $$
 of $C$ is the orbit space of ${\rm Div} _{\geq 0} ( \sim G, D)$ under this action.

 In particular, if $m < n$ then $\PP(C) = {\rm Div} _{\geq 0} ( \sim G, D)$.

 (ii) If $D_C(t) = \sum\limits _{i=0} ^{g + g^{\perp} -2} c_i t^{i}$ is Duursma's reduced polynomial of $C \subset \FF_q ^n $ and $\zeta _C(t) = \sum\limits _{i=0} ^{\infty} \cA _i (C) t^{i}$ is the $\zeta$-function of $C$ then
\begin{align*}
 c_i = \cA_i (C) = \binom{n}{d+i}^{-1} \left( \sum\limits _{\delta \in \binom{[n]}{n-d-i}} \left| {\rm Div} _{\geq 0} ( \sim (G - D_{\delta}), D) / \cL _X(G-D) \right| \right) =  \\
 \binom{n}{d+i} ^{-1} q^{ - l(G-D)} \left( \sum\limits _{\delta \in \binom{[n]}{n -d-i}}  \left|  {\rm Div} _{\geq 0} ( \sim (G - D_{\delta}), D) \right| \right) \ \ \mbox{ \rm with } \ \ 0 \leq i \leq g-1
 \end{align*}
 is the average cardinality of an $\cL_X(G-D)$-orbit space of 
 $$
 {\rm Div} _{\geq 0} ( \sim (G - D_{\delta}), D) := \{ E = G - D_{\delta} + {\rm div} (f) \geq 0 \, \vert \, f \in \FF_q (X), \, {\rm Supp} (D) \nsubseteq {\rm Supp} (E) \}  
 $$
 with $\delta \in \binom{[n]}{n-d-i}.$

 (iii)  If $D_{C^{\perp}} (t) = \sum\limits _{i=0} ^{g + g^{\perp}-2} c_i ^{\perp} t^{i}$ is Duursma's reduced polynomial of $C^{\perp}$ and $\zeta _{C^{\perp}} (t) = \sum\limits _{i=0} ^{\infty} \cA_i (C^{\perp}) t^{i}$ is the $\zeta$-function of $C^{\perp}$ then 
\begin{align*}
 c_i ^{\perp} = \cA_i (C^{\perp}) =  \\
  \binom{n}{d^{\perp}+i} ^{-1} \left( \sum\limits _{\gamma \in \binom{[n]}{d^{\perp} +i}} \left| {\rm Div} _{\geq 0} ( \sim (K_X - G + D_{\gamma}), D) / \cL _X( K_X -G) \right| \right) = \\
 \binom{n}{d^{\perp}+i} ^{-1} q^{- l (K_X -G)} \left( \sum\limits _{\gamma \in \binom{[n]}{d^{\perp}+i}} \left| {\rm Div} _{\geq 0} ( \sim (K_X - G + D_{\gamma}), D) \right| \right)  
 \end{align*}
with $0 \leq i \leq g^{\perp} -1,$  is the average cardinality of an $\cL_X(K_X -G)$-orbit space of
\begin{align*}
 {\rm Div} _{\geq 0} ( \sim (K_X - G + D_{\gamma}), D) := \\
  \{ E = K_X - G + D_{\gamma} + {\rm div} (f) \geq 0 \, \vert \, f \in \FF_q (X), \, {\rm Supp} (D) \nsubseteq {\rm Supp} (E) \} 
\end{align*}
 with $  \gamma \in \binom{[n]}{d^{\perp}+i}.$

 (iv)  The coefficients of the $\zeta$-function $\zeta _C(t) = \sum\limits _{i=0} ^{\infty} \cA _i (C) t^{i}$ of $C$ have
 $$
 \cA _i (C) \binom{n}{d+i} \in \ZZ ^{\geq 0} \ \ \mbox{  for } \ \ 0 \leq i \leq g + g^{\perp} -2 \ \ \mbox{\rm and  }
 $$
 $$
 \cA_i (C) \in \ZZ ^{\geq 0} \ \ \mbox{  for  } \ \ \forall i > g + g^{\perp} -2.
 $$
\end{corollary}

\begin{proof}

(i) First of all, one has to check that the kernel of $\cE_D : \cL _X (G) \rightarrow \cE_D \cL _X (G) = C$ equals $\cL _X (G-D)$.
If $G = G_{+} - G_{-}$ for effective divisors $G_{+}$, $G_{-}$ of $\FF_q (X)$ then $f \in \cL _X(G)$ exactly when ${\rm div} (f) _0 + G_{+} \geq {\rm div} (f) _{\infty} + G_{-}$.
Due to the disjointness of the supports of ${\rm div} (f) _0$, ${\rm div} (f) _{\infty}$ and of $G_{+}$, $G_{-}$, this is equivalent to the conditions 
${\rm div} (f) _{\infty} \leq G_{+}$ and ${\rm div} (f) _0 \geq G_{-}$.
Now, $f \in \cL _X(G)$ belongs to $\ker \cE_D$ if and only if $D \leq {\rm div} (f) _0$.
By assumption, ${\rm Supp} (G_{-}) \cap {\rm Supp} (D) \subseteq {\rm Supp} (G) \cap {\rm Supp} (D) = \emptyset$, so that the kernel of $\cE_D$ on $\cL _X(G)$ consists of the rational functions $f \in \FF_q (X)$ with ${\rm div} (f) _{\infty} \leq G_{+}$ and ${\rm div} (f) _0 \geq G_{-} - D$.
That, in turn, amounts to $G - D + {\rm div} (f) = G_{+} - G_{-} - D + {\rm div} (f) _0 - {\rm div} (f) _{\infty} \geq 0$ and reveals that $\ker \cE _D = \cL _X (G-D)$.
Now, $\cE_D : \cL _X(G) \rightarrow C = \cE_D \cL _X(G)$ restricts to a surjective map of sets
$$
\cE_D : \cL _X(G) \setminus \ker \cE_D = \cL _X(G) \setminus \cL_X (G-D) \longrightarrow  C \setminus \{ 0 _{\FF_q} ^n \}.
$$
The correspondence
$$
G + {\rm div} : \cL _X(G) \setminus \cL _X (G-D) \longrightarrow {\rm Div} _{\geq 0} ( \sim G, D), \ \ f \mapsto G + {\rm div} (f),
$$
associating to $f \in \cL _X(G) \setminus \cL _X(G-D)$ the effective divisor $G + {\rm div} $, linearly equivalent to $G$, which does not contain ${\rm Supp} (D) = \{ P_1, \ldots , P_n \}$ in its support, coincides with the quotient map of $\cL _X(G) \setminus \cL _X(G-D)$ with respect to the $\FF_q^*$-action
$$
\FF_q ^* \times ( \cL _X(G) \setminus \cL _X(G-D)) \longrightarrow ( \cL _X(G) \setminus \cL _X(G-D)), \ \ (\lambda,f) \mapsto \lambda f.
$$
Note that $\FF_q^*$ acts on $C \setminus \{ 0_{\FF_q} ^n \}$ by the rule 
$$
\FF_q^* \times (C \setminus \{ 0_{\FF_q} ^n \} ) \longrightarrow C \setminus \{ 0_{\FF_q} ^n \}, \ \ (\lambda , (c_1, \ldots , c_n)) \mapsto (\lambda c_1, \ldots , \lambda c_n)
$$
and denote by
$$
\eta _{\FF_q^*} : C \setminus \{ 0_{\FF_q} ^n \} \longrightarrow \PP (C) := C \setminus \{ 0_{\FF_q} ^n \} / \FF_q^*
$$
the projectivization map of $C$.
Straightforwardly verification establishes the $\FF_q^*$-equi\-va\-lence of
$
\cE_D : \cL _X(G) \setminus \cL _X(G-D) \longrightarrow C \setminus \{ 0_{\FF_q} ^n \},
$
i.e., $\lambda \cE_D (f) = \lambda ( f(P_1), \ldots , f(P_n)) = ( \lambda f(P_1), \ldots , \lambda f(P_n))$, $\forall f \in \cL _X(G) \setminus \cL _X (G-D)$.
Therefore $\cE_D$ induces a surjective map of finite sets
$$
\overline{\cE_D} : {\rm Div} _{\geq 0} ( \sim G, D) \longrightarrow \PP(C),
$$
closing the commutative diagram
$$
\begin{diagram}
\node{\cL_X(G) \setminus \cL_X(G-D)}   \arrow{e,t}{\cE_D}  \arrow{s,r}{G + {\rm div}}  \node{C \setminus \{ 0_{\FF_q}^n \}}  \arrow{s,r}{\eta _{\FF_q^*}}  \\
\node{{\rm Div} _{\geq 0} ( \sim G, D)}  \arrow{e,t}{\overline{\cE_D}}  \node{\PP(C)}
\end{diagram}.
$$
The corresponding fibres of $G + {\rm div}$ and $\eta _{\FF_q^*}$ are isomorphic to each other and to $\FF_q^*$, so that the fibres of $\overline{\cE_D} : {\rm Div} _{\geq 0} ( \sim G, D) \rightarrow \PP(C)$ are isomorphic to $\ker \cE_D = \cL _X(G-D)$.
For arbitrary $g \in \cL _X(G -D)$ and $G + {\rm div} (f) \in {\rm Div} _{\geq 0} ( \sim G, D)$, note that $G + {\rm div} (f+g) \in {\rm Div} _{\geq 0} ( \sim G, D)$, as far as the assumption ${\rm Supp} (D) \subseteq {\rm Supp} (G + {\rm div} (f+g))$ implies ${\rm Supp} (D) \subseteq {\rm div} (f+g) _0$, due to ${\rm div} (f + g) _{\infty} \leq G_{+}$ and ${\rm Supp} (D) \cap {\rm Supp} (G) = \emptyset$.
Then $0_{\FF_q} ^n = \cE_D ( f+g) = \cE_D (f) + \cE_D(g) = \cE_D(f)$, which contradicts $f \not \in \cL_X(G-D) = \ker \cE_D$ and verifies the correctness of the map (\ref{KernelAction}).
All fibres of (\ref{KernelAction}) are isomorphic to the linear system  $\cL _X(G-D)$, because the assumption ${\rm div} (f +g) = {\rm div} (f+h)$ for some $g, h \in \cL _X(G-D)$ requires $f +h = \lambda (f+g)$ for some $\lambda \in \FF_q ^*$ and amounts to $(\lambda -1) f = h - \lambda g \in \cL _X (G-D) = \ker \cE_D$.
The choice of $f \in \cL _X(G) \setminus \cL _X(G-D)$ specifies that $\lambda =1$ and $g=h$.
Since all the fibres of $\overline{\cE_D}$ are isomorphic to $\cL_X(G-D)$,
$$
\PP(C) = \overline{\cE_D} {\rm Div} _{\geq 0} (\sim G, D) = {\rm Div} _{\geq 0} ( \sim G, D) / \cL_X (G-D)
$$
can be viewed as the quotient space of ${\rm Div} _{\geq 0} ( \sim G, D)$, under the action of $\cL _X(G-D)$.

In particular, if $m <n$ then $\deg (G-D) <0$ and $\cL_X(G-D) = \{ 0_{\FF_q(X)} \}$.
That allows to identify $\PP(C) = {\rm Div} _{\geq 0} ( \sim G, D)$ with the effective divisors, linearly equivalent to $G$, which do not contain 
$\{ P_1, \ldots , P_n \}$ in its support.

(ii)  Note that the support function ${\rm Supp} : C \rightarrow \{ 0,1, \ldots , n \}$ is invariant under the action
$$
\FF_q^* \times C \longrightarrow C, \ \ (\lambda, (c_1, \ldots , c_n)) \mapsto (\lambda c_1, \ldots , \lambda c_n)
$$
of $\FF_q^*$ and descends to ${\rm Supp} \PP(C) \rightarrow \{ 1, \ldots , n \}$ with ${\rm Supp} ([a]) = {\rm Supp} \eta _{\FF_q^*} (a) = {\rm Supp} (a)$ for $\forall a \in C \setminus \{ 0_{\FF_q} ^n \}$.
In particular, 
$$
\eta _{\FF_q^*} : \left( C \setminus \{ 0_{\FF_q} ^n \} \right) ^{( \subseteq \gamma)} \longrightarrow \PP(C) ^{( \subseteq  \gamma )} := \{ [a] \in \PP(C) \, \vert \, {\rm Supp} ([a]) \subseteq \gamma \}
$$
is an unramified $\FF_q^*$-Galois covering and (\ref{AverageC}) takes the form
$$
c_i = \binom{n}{d+i} ^{-1} \left( \sum\limits _{\gamma \in \binom{[n]}{d+i}} \left| \PP(C) ^{( \subseteq \gamma)} \right| \right) \ \ \mbox{  for  } \ \ 
\forall 0 \leq i \leq g-1,
$$
according to $\left| \left( C \setminus \{ 0_{\FF_q} ^n \} \right) ^{( \subseteq \gamma )} \right| = (q-1) \left| \PP(C) ^{( \subseteq \gamma )} \right|$.
If $\delta := \neg \gamma = [n] \setminus \gamma \in \binom{[n]}{n-d-i}$ is the complement of $\gamma \in \binom{[n]}{d+i}$, it suffices to show that
\begin{equation}   \label{ProjectivizatonSupportSet}
\PP(C) ^{( \subseteq \gamma)} \simeq {\rm Div} _{\geq 0} ( \sim (G - D_{\delta} ), D) / \cL_X(G-D)
\end{equation}
is the orbit space of 
$$
{\rm Div} _{\geq 0} ( \sim (G - D_{\delta}), D) := \{ E = G - D_{\delta} + {\rm div} (f) \geq 0 \, \vert \, f \in \FF_q(X), \, {\rm Supp} (D) \nsubseteq {\rm Supp} (E) \}
$$
under the action
$$
\cL_X(G-D) \times {\rm Div} _{\geq 0} ( \sim (G - D_{\delta}), D) \longrightarrow {\rm Div} _{\geq 0} ( \sim (G - D_{\delta}), D), 
$$
$$
(g, G - D_{\delta} + {\rm div} (f)) \mapsto G - D_{\delta} + {\rm div} (f+g)
$$
of $\cL _X(G-D)$, in order to conclude the proof of (ii).
To this end, note that 
\begin{align*}
\cE_D ^{-1} \left( C ^{( \subseteq \gamma)} \setminus \{ 0_{\FF_q} ^n \} \right) = \{ f \in \cL _X(G) \setminus \cL _X(G-D) \, \vert \, {\rm Supp} ( D_{\delta}) \leq {\rm div} (f) _0 \} = \\
 \cL _X (G - D_{\delta}) \setminus \cL _X (G-D).
\end{align*}
The considerations from (i), applied to $G - D_{\delta}$ instead of $G$ provide the commutative diagram
$$
\begin{diagram}
\node{\cL_X (G - D_{\delta}) \setminus \cL _X (G-D)}  \arrow{s,r}{G - D_{\delta} + {\rm div}}  \arrow{e,t}{\cE_D}  \node{ \left( C \setminus \{ 0_{\FF_q} ^n \} \right) ^{( \subseteq \gamma)} }  \arrow{s,r}{\eta _{\FF_q^*}}  \\
\node{{\rm Div} _{\geq 0} ( \sim (G - D_{\delta}), D)}  \arrow{e,t}{\overline{\cE _D}}  \node{\PP(C) ^{( \subseteq \gamma)}}
\end{diagram},
$$
where $\overline{\cE_D}$ is a surjective map, whose fibres are isomorphic to $\cL_X(G-D)$.
That allows the identification (\ref{ProjectivizationSupportSet}) with $\left| {\rm Div} _{\geq 0} ( \sim (G - D_{\delta}), D) \right| = q^{l(G-D)} \left| \PP(C) ^{( \subseteq \gamma)} \right|$.

(iii) follows from (\ref{AverageCPerp}) by observing that
\begin{align*}
\PP(C^{\perp}) ^{( \subseteq \gamma)} = 
 \PP ( \cE_D \cL_X (K_X - G+D) ) ^{( \subseteq \gamma )} \simeq  \\
{\rm Div} _{\geq 0} ( \sim K_X -G +D - (D - D_{\gamma})), D) / \cL_X(G-D) =  \\ {\rm Div} _{\geq 0} ( \sim (K_X - G + D_{\gamma}), D) / \cL _X(G-D)  \ \ \mbox{ for} \ \  \forall \gamma \in \binom{[n]}{d^{\perp} +i}.
\end{align*}

(iv) Note that (ii) implies $\cA _i (C) \binom{n}{d+i} \in \ZZ ^{\geq 0}$ for $\forall 0 \leq i \leq g-1$ and (iii) guarantees that $\cA_i (C^{\perp}) \binom{n}{d^{\perp} +i} \in \ZZ ^{\geq 0}$ for $\forall 0 \leq i \leq g^{\perp} -1$.
Making use of (\ref{MWA}), one concludes that
$$
\cA_i (C) \binom{n}{d+i} = q^{i-g+1} \cA_{g + g^{\perp} -2-i} (C^{\perp}) \binom{n}{d+i} + \binom{n}{d+i} \left( \frac{q^{i-g+1} -1}{q-1} \right)^\in \ZZ ^{\geq 0}  
$$
for $\forall g \leq i \leq g + g{\perp} -2,$ according to $d^{\perp} + (g + g^{\perp} -2-i) = n-d-i$ and $\binom{n}{n-d-i} = \binom{n}{d+i}$.
In the case of $i > g + g^{\perp} -2$, (\ref{MWA}) reduces to
$$
\cA_i (C) = \frac{q^{i-g+1} -1}{q-1} \in \ZZ ^{\geq 0}.
$$

\end{proof}

By Theorem 1.1.28 and Exercise 1.1.29 from \cite{TVN}, the homogeneous weight enumerator of an $\FF_q$-linear code $C \subset \FF_q ^n$ with dual $C^{\perp}$ of minimum distance $d^{\perp} \geq 2$ can be expressed in the form
$$
\cW_C (x,y) = x^n + \sum\limits _{i=0} ^{n-d} B_i (x-y)^{i} y^{n-i} \ \ \mbox{  with } \ \ B_i = (q-1) \left( \sum\limits _{\alpha \in \binom{[n]}{i}} \left| \PP (C) ^{( \subseteq \neg \alpha)} \right| \right).
$$
Corollary \ref{CoeffDCAreAverageLCodes} reveals that  Tsfasman-Vl${\rm \breve{a}}$dut-Nogin's coefficients $B_{d+i} = \binom{n}{d+i} (q-1) c_i$ are closely related with the coefficients $c_i$ of Duursma's reduced polynomial $D_C(t)$ of $C$ for $\forall 0 \leq i \leq g-1$.

\begin{proposition}    \label{ProbabilisticInterpretationsCoeffDCDCPerp}
Let $(C, +) < (G^n, +)$ be an additive code  of minimum distance $d \geq 2$ and genus $g \geq 1$ with dual $(C^{\perp},.) < (\widehat{G}^n, .)$ of minimum distance $d^{\perp} \geq 2$ and genus $g^{\perp} \geq 1$.
Suppose that $D_C(t) = \sum\limits _{i=0} ^{g + g^{\perp} -2} c_i t^{i}$, $D_{C^{\perp}} (t)  = \sum\limits _{i=0} ^{g + g^{\perp} -2} c_i ^{\perp} t^{i} \in \QQ[t]$ are Duursma's reduced polynomials of $C$, $C^{\perp}$, $p^{(s)} _C$ (respectively, $\pi ^{(s)}_{C^{\perp}}$) are the probabilities of $a \in G^n$  (respectively, $\pi \in \widehat{G}^n$) of weight ${\rm wt} (a) =s$ (respectively, ${\rm wt} (\pi ) =s$) to belong to $C$ (respectively, to $C^{\perp}$) and
$\overline{p_a^{(s)}}$ (respectively, $\overline{p_{\pi} ^{(s)}}$) are the probabilities of $\gamma \in \binom{[n]}{s}$ to contain the support ${\rm Supp} (a)$ of $a \in C$ (respectively, the support ${\rm Supp} ( \pi)$ of $\pi \in C^{\perp}$).
Then
\begin{equation}    \label{LowerProbC}
{\rm (i)} \ \ c_i = \sum\limits _{s=d} ^{d+i} p_C ^{(s)} \binom{d+i}{s} (|G|-1) ^{s-1} \ \ \mbox{\rm for } \ \ \forall 0 \leq i \leq g-1,
\end{equation}
\begin{equation}    \label{UpperProbC}
c_i = |G|^{i-g+1} \left[ \sum\limits _{s = d^{\perp}} ^{n-d-i} p_{C^{\perp}} ^{(s)}  \binom{n-d-i}{s}  (|G|-1) ^{s-1} \right] \ \ \mbox{\rm  for } \ \
 \forall g \leq i \leq g + g^{\perp} -2;
\end{equation}
\begin{equation}   \label{LowerProbBarC}
{\rm (ii)} \ \ c_i = (|G|-1)^{-1} \left( \sum\limits _{a \in C \setminus \{ 0_G ^n \}} \overline{p_a ^{(d+i)}} \right) \ \ \mbox{\rm  for  }
 \ \ \forall 0 \leq i \leq g-1,
\end{equation}
\begin{equation}   \label{UpperProbBarC}
c_i = (|G|-1) ^{-1} |G|^{i-g+1} \left( \sum\limits _{\pi \in C^{\perp} \setminus \{ \varepsilon ^n \}} \overline{p _{\pi} ^{(n-d-i)}} \right) \ \
\mbox{\rm for } \ \ \forall g \leq i \leq g + g^{\perp} -2.
\end{equation}
\end{proposition}

\begin{proof}

(i) If $(G^n)^{(s)} := \{ a \in G^n \, \vert \, {\rm wt} (a) = s \}$, respectively, $C^{(s)} := \{ a \in C \, \vert \, {\rm wt} (a) =s \}$ are the subsets of the words of weight $1 \leq s \leq n$, $q := |G|$ and $\cW _C ^{(s)} := |C ^{(s)}|$ then the probability of $a \in (G^n) ^{(s)}$ to belong to $C$ is
$$
p_C ^{(s)} = \frac{|C ^{(s)}|}{|(G^n) ^{(s)}|} = \frac{\cW _C ^{(s)}}{\binom{[n]}{s} (q-1)^s}.
$$
According to (\ref{DuursmasCoefficientsC}) for Proposition \ref{WCByDCAnd CoeffDenominators}, if $0 \leq i \leq g-1 = n-d-k$ then
\begin{equation}   \label{LowerDuursmasCoefficientsC}
(q-1) \binom{n}{d+i} c_i = \sum\limits _{s=d} ^{d+i} \cW _C ^{(s)} \binom{n-s}{n-d-i}
\end{equation}
as far as $\cW _C ^{(0)} = \cM _{n, n+1-k} ^{(0)} = 1$, $\cW _C ^{(s)} =0$ for $\forall 1 \leq s \leq d-1$ and $\cM _{n, n+-k} ^{(s)} =0$ for
 $\forall 1 \leq s \leq d+i \leq n-k$.
Substituting $\cW_C ^{(s)} = \binom{n}{s} (q-1) ^s p_C ^{(s)}$ in (\ref{LowerDuursmasCoefficientsC}) and making use of
$
\binom{n}{d+i} ^{-1} \binom{n}{s} \binom{n-s}{n-d-i} = \binom{d+i}{s},
$
one concludes that (\ref{LowerProbC}).

The application of (\ref{LowerProbC}) to $C^{\perp}$ yields
$$
c_i ^{\perp} = \sum\limits _{s = d^{\perp}} ^{d^{\perp} +i} p_{C^{\perp}} ^{(s)} \binom{d^{\perp} +i}{s} (q-1)^{s-1} \ \ \mbox{ for } \ \
\forall 0 \leq i \leq g^{\perp} -1.
$$
According to (\ref{DCFormula}) from Corollary (\ref{LowerPartD}), Duursma's reduced polynomial $D_C(t)$ can be represented as
$$
D_C(t) = \sum\limits _{i=0} ^{g -1} c_i t^{i} + \sum\limits _{i=g} ^{ g + g^{\perp} -2} c^{\perp} _{g + g^{\perp} -2-i} q^{i -g+1} t^{i}.
$$
Therefore
\begin{equation}   \label{CoeffExpressionCorollary5}
c_i = q^{i -g+1} c^{\perp} _{g + g^{\perp} -2-i} \ \ \mbox{  for } \ \ \forall g \leq i \leq g + g^{\perp} -2.
\end{equation}
Plugging in (\ref{LowerProbBarC}) with $0 \leq g + g^{\perp} -2-i = n-d-d^{\perp} -i \leq g^{\perp} -2$ in the above formula, one obtains (\ref{UpperProbC}).

(ii)  If $a \in C$ has support ${\rm Supp} (a) \in \binom{[n]}{s}$ and $s \leq w \leq n$, then the number of $\gamma \in \binom{[n]}{w}$, containing ${\rm Supp} (a)$ equals $\binom{n-s}{w-s}$.
Thus, the probability of $\gamma \in \binom{[n]}{w}$ to contain ${\rm Supp} (a)$ equals $\overline{p_a ^{(w)}} = \frac{\binom{n-s}{w-s}}{\binom{n}{w}}$.
If ${\rm wt} (a) = s > w$ then $\overline{p_a ^{(w)}} =0$.
Making use of this, one represents (\ref{LowerDuursmasCoefficientsC}) as
\begin{align*}
c_i =
 \sum\limits _{s=1} ^{d+i} \frac{\cW_C ^{(s)} \binom{n-s}{d+i-s}}{(q-1) \binom{n}{d+i}} =
  \sum\limits _{s=1} ^{d+i} \sum\limits _{a \in C^{(s)}} \frac{\overline{p _a ^{(d+i)}}}{q-1} =
  (q-1) ^{-1} \left( \sum\limits _{a \in C \setminus \{ 0_g ^n \}} \overline{p_b ^{(d^{\perp} +i)}} \right)
  \end{align*}
for $\forall 0 \leq i \leq g-1.$
Combining with (\ref{CoeffExpressionCorollary5}), one derives (\ref{UpperProbBarC}).

\end{proof}

In the case of $\FF_q$-linear codes,  Proposition \ref{ProbabilisticInterpretationsCoeffDCDCPerp} specializes to the following

\begin{corollary}   \label{ProbabilisticInterpretationsCoeffDCDCPerpForLinearCodes}
Let $C \subset \FF_Q ^n$ be an $\FF_q$-linear code of minimum distance $d \geq 2$ and genus $g \geq 1$ with dual $C^{\perp} \subset \FF_q ^n$ of minimum distance $d^{\perp} \geq 2$ and genus $g^{\perp} \geq 1$.
Denote by $D_C(t) = \sum\limits _{i=0} ^{g + g^{\perp} -2} c_i t^{i}$, $D_{C^{\perp}} (t) = \sum\limits _{i=0} ^{g + g^{\perp} -2} c_i ^{\perp} t^{i} \in \QQ[t]$ Duursma's reduced polynomials of $C$, $C^{\perp}$, put $\pi ^{(s)} _{\PP(C)}$ (respectively, $\pi ^{(s)} _{\PP ( C^{\perp})}$) for the probability of $[a] \in \PP^{n-1} ( \FF_q)$ to belong to $\PP(C)$ (respectively, to $\PP (C^{\perp})$) and designate by $\overline{\pi _{[a]} ^{(s)}}$ (respectively, by
$\overline{\pi _{[b]} ^{(s)}}$)  the probability of $\gamma \in \binom{[n]}{s}$ to contain the support of $[a] \in \PP(C)$ (respectively, of $[b] \in \PP (C^{\perp})$).
Then:
$$
{\rm (i)} \ \ c_i = \sum\limits _{s=d} ^{d+i} \pi _{\PP(C)} ^{(s)} \binom{d+i}{s} (q-1) ^{s-1} \ \ \mbox{\rm  for } \ \ \forall 0 \leq i \leq g-1,
$$
$$
c_i = q^{i-g+1} \left[ \sum\limits _{s = d^{\perp}} ^{n-d-i} \pi ^{(s)} _{\PP(C^{\perp})} \binom{n-d-i}{s} (q-1) ^{s-1} \right] \ \ \mbox{\rm  for } \ \
\forall g \leq i \leq g + g^{\perp} -2;
$$
$$
{\rm (ii)} \ \ c_i = \sum\limits _{[a] \in \PP(C)} \overline{\pi _{[a]} ^{(d+i)}} \ \ \mbox{\rm for } \ \ \forall 0 \leq i \leq g-1,
$$
$$
c_i = q^{i-g+1} \left( \sum\limits _{[b] \in \PP (C^{\perp})} \overline{\pi _{[b]} ^{(n-d-i)}} \right) \ \ \mbox{\rm for } \ \
\forall  g \leq i \leq g + g^{\perp} -2.
$$
\end{corollary}

\begin{proof}

(i)  It suffices to note that $\PP(C) ^{(s)} := \{ [a] \in \PP(C) \, \vert \, {\rm wt} ([a]) =s \}$ is of cardinality
$$
\left| \PP(C) ^{(s)} \right| = \frac{|C^{(s)}|}{|\FF_q ^*|} = \frac{\cW_C ^{(s)}}{q-1} \ \ \mbox{  for } \ \ \forall 1 \leq s \leq n
$$
and $\PP^{n-1} ( \FF_q) ^{(s)} := \{ [a] \in \PP^{n-1} ( \FF_q) \, \vert \, {\rm wt} (a) = s \}$  is of cardinality
\begin{align*}
\left| \PP ^{n-1} ( \FF_q) ^{(s)} \right| = \binom{n}{s} \frac{| ( \FF_q ^* )^s|}{|\FF_q ^* |} = \binom{n}{s} (q-1)^{s-1},
\end{align*}
so that
$$
\pi ^{(s)} _{\PP(C)} = \frac{| \PP(C) ^{(s)}|}{| \PP^{n-1} (\FF_q) ^{(s)} |} = \frac{\cW _C ^{(s)}}{\binom{n}{s} (q-1)^s} = p_C ^{(s)}
$$
and (i) is an immediate consequence of (\ref{LowerProbC}) and (\ref{UpperProbC}) from Proposition \ref{ProbabilisticInterpretationsCoeffDCDCPerp}.

(ii)  The first equality follows from  (\ref{LowerProbBarC}) by noting that the support of $a \in C \setminus \{ 0 _{\FF_q} ^n \}$ is constant along an $\FF_q ^*$-orbit on $C \setminus \{ 0_{\FF_q} ^n \}$ and the projectivization $\PP(C) := C \setminus \{ 0_{\FF_q} ^n \} / \FF_q ^*$ of $C$ is the $\FF_q^*$-orbit space of $C \setminus \{ 0_{\FF_q} ^n \}$.
The second equality follows from (\ref{UpperProbBarC}), $\PP(C^{\perp}) := C^{\perp} \setminus \{ 0_{\FF_q} ^n \} / \FF_q ^*$ and the fact that the weight is constant along the $\FF_q^*$-orbits on $C^{\perp} \setminus \{ 0_{\FF_q} ^n \}$.

\end{proof}

 %% \newpage
%%%%%%%%%%%%%%%%%%%%%%%%%%%%%%%%%%%%%%%%%%%%%%%%%%%%%%%%%%%%%%%%%%%%%%%%%%%%%%%%%%%%%%%%%%%%%%%%%%%%%%%%%%%%%%%%%%%%%%%%%%%%%%%%%%%%%%%%%%


\begin{thebibliography}{99}

\bibitem{ByrneGreferathSullivan}
E. Byrne, M. Greferath, M. E. O'Sullivan,
The linear programming bound for codes over finite Frobenius rings,
\emph{ Designes, Codes and Cryptography}, \textbf{ 42} (2007), 289-301.


\bibitem{Delsarte}
P. Delsarte, Bounds for unrestricted codes by linear programming, \emph{Philips Res. Rep.} \textbf{27}  (1972), 272-289.



\bibitem{Delsarte6}
P. Delsarte, An algebraic approach to the association schemes of coding theory, \emph{ Philips Res. Repts. Suppl}, \textbf{10}, 1973.


\bibitem{D1}
  I. Duursma,  Weight distribution of geometric Goppa codes,
     \emph{Transections of the American Mathematical Society}, \textbf{351} (1999),   3609--3639.







\bibitem{D2}
  I. Duursma,
     From weight enumerators to zeta functions,
     \emph{Discrete Applied Mathematics}, \textbf{ 111} (2001),    55–-73.



\bibitem{Heide}
H. Gluesing-Luersen, Fourier-reflexive partitions and Mac Williams identities for additive codes,
 \emph{Designs, Codes and Cryptography}, \textbf{75} (2015), 543-563.


\bibitem{GreferathSchmidt}
M. Greferath, S. E. Schmidt,
Finite-ring combinatorics and Mac Williams's equivalence theorem,
\emph{ Journal of Combinatorics Theory}, Ser. A \textbf{92} (2000), 17-28.


\bibitem{HonoldLandjev}
T. Honold, I. Landjev,
Mac Williams identities for linear codes over finite Frobenius rings,
in \emph{Proceedings of the Fifth International Conference on Finite Fields and Applications},
(Augsburg 1999), Jungnickel D. and Niederreiter H. edited, Springer, 2001, 276-292.


\bibitem{HuffmanPless}
W. C. Huffman, V. Pless, \emph{Fundamentals of Error Correcting Codes,}
Cambridge University Press, 2003.

\bibitem{LuKumarYang}
H. Lu, P. V. Kumar, E. Yang, On the input-output weight enumerators of product accumulate codes,
\emph{ IEEE Communication Letters}, \textbf{8} (2004), 520-521.


\bibitem{OriginalMW}
F. J. Mac Williams, Combinatorial problems of elementary abelian groups, PhD Thesis, Harvard University, 1962.


\bibitem{KM1}
     A. Kasparian, I. Marinov,
   Duursma's reduced polynomial,
  arXiv:1505.01993v1[cs.IT] 8 May 2015.


  \bibitem{KM2}
  A. Kasparian, I.Marinov, Riemann Hypothesis Analogue for locally finite modules
over the absolute Galois group of a finite field, arXiv:1608.05328v1[math.AG] 18 August 2016.


\bibitem{K4}
A. Kasparian, Algebraic-geometric families of linear codes, in progress.



  \bibitem{KhamyMcEliece}
  M. El-Khamy, R. J. Mc Eliece, The partition weight enumerator of MDS codes and its applications,
  in \emph{ Proceedings of the IEEE International Symposium on Information Theory ISIT 2005}
  (Adelaide, Australia), 2005, 926-930.

  \bibitem{NX}
H.  Niederreiter, Ch. Xing, \emph{Algebraic Geometry in Coding Theory and Cryptography}, Princeton University Press, 2009.


\bibitem{PShW}
R. Pellikaan, B.-Z. Shen, G. J. M. van Wee,
Which linear codes are algebraic geometric?, IEEE Trans. Inform. Theory \textbf{IT-37} (1991),  583-602.

  \bibitem{Hugues}
    H. Randriambololona,
 Harder-Narasimhan theory for linear codes,
  arXiv:1609.00738v1[cs.CO]  2 Sept 2016


 \bibitem{TVN}
    M. Tsfasman, S. Vl${\rm \breve{a}}$dut, D. Nogin,
  \emph{Algebraic Geometry Codes:  Basic Notions, }
        Providence, RI: American Mathematical Society, 2007.


\bibitem{Wood}
J. Wood, Lecture notes on theMac Williams identities and the extension theorem,
in \emph{ Proceedings of the CIMAT International School and Conference on Coding Theory},
 CIMAT, Guanajuato, Mexico, November 28-December 4, 2008.


 \bibitem{ZinovievEricson}
 V. A. Zinoviev, T. Ericson,
 On Fourier invariant partitions of finite abelian groups and the Mac Williams identity for group codes,
 \emph{ Problems of Information Transmission} \textbf{32} (1996), 117-122.




\end{thebibliography}
 \end{document}